\documentclass[nocopyrightspace,preprint,10pt]{sigplanconf}
\usepackage{alltt,fancyvrb,url}
\usepackage{array}
\usepackage{amsmath}
\usepackage{amssymb}
\usepackage{amsthm}
\usepackage{mathpartir}
\usepackage{multirow}
\usepackage{natbib}
\usepackage{upgreek}
\usepackage{graphicx}

\newif\iffull
\fulltrue

\usepackage{color}
\usepackage{calc}
\definecolor{gray}{rgb}{0.9,0.9,0.9}
\definecolor{red}{rgb}{1,0,0}
\newcommand{\graybox}[1]{\mbox{\setlength{\fboxsep}{0.5pt}%
    \colorbox{gray}{$#1$}}}

\usepackage{xspace}
\newcommand{\ma}[1]{\ensuremath{#1}\xspace}

\newcommand\asblm[3]{\ablm{#1}{#2}{#3}{\uplambda}{#3}}

\newcommand\allocname{\ensuremath{\mathsf{alloc}}}

\newcommand\mans{\ensuremath{A}}
\newcommand\manso{\ensuremath{A'}}

\newcommand\opaque{\bullet}
\newcommand\topaque{\bullet^\mtyp}

\newcommand{\dotcup}{\mathrel{\ensuremath{\mathaccent\cdot\cup}}}

\newcommand\with[2]{{#1}/{#2}}
\newcommand\remcon[2]{\with{#1}{#2}}

\newcommand\mtyp{\ensuremath{T}}
\newcommand\mbasetyp{\ensuremath{B}}
\newcommand\menv{\ensuremath{\rho}}
\newcommand\menvo{\ensuremath{\varrho}}
\newcommand\mprg{\ma{P}}
\newcommand\mprgo{\ensuremath{Q}}
\newcommand\mexp{\ensuremath{E}}
\newcommand\mexpo{\ensuremath{E'}}

\newcommand\mvmod{{\ensuremath{\vec{\mmod}}}}
\newcommand\mvmodo{{\ensuremath{\vec{\mmodo}}}}

\newcommand\mmod{\ensuremath{M}}
\newcommand\mmodo{\ensuremath{N}}
\newcommand\mmodvar{\ensuremath{f}}
\newcommand\mmodvaro{\ensuremath{g}}
\newcommand\mmodvaroo{\ensuremath{h}}
\newcommand\mcon{\ensuremath{C}}
\newcommand\mconset{\ensuremath{\mathcal{C}}}
\newcommand\mcono{\ensuremath{D}}

\newcommand\mval{\ensuremath{V}}
\newcommand\mpreval{\ensuremath{U}}
\newcommand\mvalo{\ensuremath{U}} 
\newcommand\mvar{\ensuremath{X}}
\newcommand\mvaro{\ensuremath{y}}

\newcommand\mnum{\ensuremath{n}}
\newcommand\mnumo{\ensuremath{m}}

\newcommand\mctx{\ensuremath{\mathcal{E}}}

\newcommand\maddr{\ensuremath{a}}
\newcommand\maddro{\ensuremath{b}}

\newcommand\msto{\ensuremath{\sigma}}
\newcommand\mstoo{\ensuremath{\varsigma}}

\newcommand\mstate{\ensuremath{\varsigma}}

\newcommand\mop{O}
\newcommand\moppred{o?}

\newcommand\mlab{\ell}
\newcommand\mlabo{\ell'}

\usepackage{pifont}

\newcommand\proves[2]{{\vdash{#1}:{#2}\mbox{\,\ding{51}}}}
\newcommand\refutes[2]{\vdash{#1}:{#2}\mbox{\,\ding{55}}}
\newcommand\ambig[2]{\vdash{#1}:{#2}\,\mbox{\bf ?}} 

\newcommand\sbegin[2]{#1;\ #2}
\newcommand\sreclamnp[3]{\ensuremath{\uplambda_{#2} #1.#3}}
\newcommand\slamnp[2]{\sreclamnp{#1}{\relax}{#2}}
\newcommand\smod[3]{\ensuremath{(\syntax{module}\:#1\:#2\:#3)}}
\newcommand\stlam[3]{\slamnp{#1\!:\!#2}{#3}}
\newcommand\slam[2]{(\slamnp{#1}{#2})}
\newcommand\sreclam[3]{(\sreclamnp{#1}{#2}{#3})}
\newcommand\strec[3]{\srec{{#1}\!:\!{#2}}{#3}}
\newcommand\srec[2]{\mu{#1}.{#2}}
\newcommand\sapp[2]{\ensuremath{#1\:#2}}

\newcommand\sif[3]{\mathsf{if}\:#1\:#2\:#3}

\newcommand\sany{\ensuremath{\syntax{any}}}
\newcommand\sarr[2]{\ensuremath{#1\!\mapsto\! #2}}
\newcommand\sdep[3]{\ensuremath{#1\!\mapsto\! \uplambda#2.#3}}
\newcommand\stdep[4]{\sdep{#1}{#2\!:\!#3}{#4}}
\newcommand\spred[1]{\ensuremath{\liftpred{#1}}}
\newcommand\scons[2]{(\syntax{cons}\:#1\:#2)}
\newcommand\sconsop{\syntax{cons}}

\newcommand\szero{\syntax{zero?}}
\newcommand\ssucc{\syntax{add1}}
\newcommand\snatp{\syntax{nat?}}
\newcommand\sequalp{\syntax{=}}
\newcommand\splus{\syntax{+}}
\newcommand\soptwo[3]{{#2}({#1},{#3})}
\newcommand\sopone[2]{{#1}({#2})}

\newcommand\sfalse{\syntax{ff}}
\newcommand\strue{\syntax{tt}}
\newcommand\sboolp{\syntax{bool?}}
\newcommand\sempty{\syntax{empty}}
\newcommand\semptyp{\syntax{empty?}}
\newcommand\scar{\syntax{car}}
\newcommand\scdr{\syntax{cdr}}
\newcommand\sprocp{\syntax{proc?}}
\newcommand\sfalsep{\syntax{false?}}
\newcommand\struep{\syntax{true?}}
\newcommand\sconsp{\syntax{cons?}}

\newcommand\syntax[1]{\ma{\mathsf{#1}}}

\newcommand\sand[2]{{#1}\wedge{#2}}
\newcommand\sor[2]{{#1}\vee{#2}}

\newcommand\sconsc[2]{\langle#1,\!#2\rangle}
\newcommand\sandc[2]{#1\wedge#2}
\newcommand\sorc[2]{#1\vee#2}
\newcommand\srecc[2]{\mu #1.#2}

\newcommand\achk[6]{\chk{#1}{#3}{#4}{#5}{#2}}
\newcommand\chk[5]{\ensuremath{\syntax{mon}^{#2,#3}_{#4}(#1, #5)}}
\newcommand\achksimple[2]{\ensuremath{\syntax{mon}(#1,{#2})}}
\newcommand\ablm[5]{\ensuremath{\syntax{blame}^{#1}_{#2}}}
\newcommand\simpleblm[2]{\ablm{#1}{#2}\relax\relax\relax}

\newcommand\sblm[2]{\ensuremath{\syntax{blame}^{#1}_{#2}}}

\newcommand\stypbool{\syntax{B}}
\newcommand\stypnum{\syntax{N}}
\newcommand\starr[2]{{#1}\rightarrow{#2}}
\newcommand\stcon[1]{\syntax{con}({#1})}
\newcommand\sflat[1]{\ensuremath{\syntax{flat}(#1)}}

\newcommand\depbless[7]{
  \ma{\slam{#2}{\achk{#3}{(\sapp{#7}{\achk{#1}{#2}{#5}{#4}{#6}{}})}{#4}{#5}{#6}{}}}}

\newcommand\simpledepbless[4]{
  \ma{\slam{#2}{\achksimple{#3}{(\sapp{#4}{\achksimple{#1}{#2}})}}}}

\newcommand\deltamap[3]{\delta({#1},{#2})\ni{#3}}
\newcommand\absdeltamap[3]{\absdelta({#1},{#2})\ni{#3}}

\newcommand\alloc[1]{\ensuremath{\mathsf{alloc}(#1)}}

\newcommand\subst[3]{\ensuremath{[#1/#2]#3}}

\newcommand\sstep{\mathbf{s}}       
\newcommand\stdstep{\longmapsto}
\newcommand\multistdstep{\longmapsto\!\!\!\!\!\rightarrow}

\newcommand\unload{\ensuremath{\mathcal{U}}}

\newcommand\mstep{\ensuremath{\stdstep\quad}}
\newenvironment
    {machine}[2][\relax]
    {\begin{display*}[#1]{#2}{\textwidth}\[
\begin{tabular}{@{}>{$}p{2.1in}<{$}>{$}p{.2in}<{$}>{$}p{2.2in}<{$}>{$}p{1.8in}<{$}}
}
    {\end{tabular}\]\end{display*}}   

\newcommand\plaindelta{\cn\delta}
\newcommand\absdelta{\delta}

\newcommand\liftpred[1]{\sflat{#1}\relax}

\newcommand\projleft[1]{\pi_1{#1}}
\newcommand\projright[1]{\pi_2{#1}}
\newcommand\proj\pi

\newcommand\fc{\textsc{fc}} 

\newcommand\vcons[2]{({#1},{#2})}

\newcommand\s[4]{\langle{#1},{#2},{#4},{#3}\rangle}
\newcommand\se[1]{\s{#1}\menv\msto\mcont}
\newcommand\sk[1]{\s\mval\menv\msto{#1}}

\newcommand\mcont\kappa
\newcommand\mconto\iota

\newcommand\cont[1]{\textsf{#1}}

\newcommand\ab[1]{\widehat{#1}}
\newcommand\cn[1]{\widetilde{#1}}

\newcommand\kchk[6]{\cont{chk}^{#3,#4}_{#5}({#1},{#2},{#6})}
\newcommand\kchkor[8]{\cont{chk-or}^{#5,#6}_{#7}({#1},{#2},{#3},{#4},{#8})}

\newcommand\kchkconso[8]{\cont{chk-cons}^{#5,#6}_{#7}({#1},{#2},{#3},{#4},{#8})}
\newcommand\kfn[4]{\cont{fn}^{#3}({#1},{#2},{#4})}
\newcommand\kap[4]{\cont{ar}^{#3}({#1},{#2},{#4})}
\newcommand\kif[4]{\cont{if}({#1},{#2},{#3},{#4})}
\newcommand\kopone[3]{\cont{op}^{#2}({#1},{#3})}
\newcommand\koptwo[5]{\cont{opr}^{#4}({#1},{#2},{#3},{#5})}
\newcommand\koptwol[5]{\cont{opl}^{#4}({#1},{#2},{#3},{#5})}
\newcommand\kdem[3]{\cont{begin}({#1},{#2},{#3})}
\newcommand\kmt{\cont{mt}}

\newcommand\mkaddr{k}
\newcommand\mkaddro{i}

\newcommand{\hbra}{
\hbox to .995 \columnwidth{\vrule width0.3mm height 1.8mm depth-0.3mm
                    \leaders\hrule height1.8mm depth-1.5mm\hfill
                    \vrule width0.3mm height 1.8mm depth-0.3mm}}
\newcommand{\hket}{
\hbox to .995 \columnwidth{\vrule width0.3mm height1.5mm
                    \leaders\hrule height0.3mm\hfill
                    \vrule width0.3mm height1.5mm}}

\newcommand{\hbraF}{
\hbox to .995 \textwidth{\vrule width0.3mm height 1.8mm depth-0.3mm
                    \leaders\hrule height1.8mm depth-1.5mm\hfill
                    \vrule width0.3mm height 1.8mm depth-0.3mm}}
\newcommand{\hketF}{
\hbox to .995 \textwidth{\vrule width0.3mm height1.5mm
                    \leaders\hrule height0.3mm\hfill
                    \vrule width0.3mm height1.5mm}}

\newcommand{\displayCap}[1]{\textbf{#1}}

\newenvironment{display}[2][]{
  \vskip 12pt
  \noindent 
  \begin{minipage}{\columnwidth}
    \displayCap{#2}\hspace{\stretch{1}}\textit{#1\:}\\[-0.6ex]
    \hbra\\[-3ex]
}{
    \\[-3ex] \hket\\[-1.5ex]
  \end{minipage}
}

\newenvironment{displayfig}[2][]{
  \noindent 
  \begin{minipage}{\columnwidth}
    \displayCap{#2}\hspace{\stretch{1}}\textit{#1\:}\\[-0.6ex]
    \hbra\\[-3ex]
}{
    \\[-3ex] \hket\\[-1.5ex]
  \end{minipage}
}

\newenvironment{display*}[3][]{
  \vskip 4pt
  \noindent 
  \begin{minipage}{#3}
    \displayCap{#2}\hspace{\stretch{1}}\textit{#1\:}\\[-0.6ex]
    \hbraF\\[-3ex]
}{
    \\[-3ex] \hketF\\[-1.5ex]
  \end{minipage}
}

  \renewenvironment{thebibliography}[1]{%
    \begin{oldthebibliography}{#1}%
      \setlength{\parskip}{0ex}%
      \setlength{\itemsep}{0ex}%
  }%
  {%
    \end{oldthebibliography}%
  }

\newcommand{\squishlist}{
 \begin{list}{$\bullet$}
  { \setlength{\itemsep}{0pt}
     \setlength{\parsep}{1pt}
     \setlength{\topsep}{1pt}
     \setlength{\partopsep}{0pt}
     \setlength{\leftmargin}{1.5em}
     \setlength{\labelwidth}{1em}
     \setlength{\labelsep}{0.5em} } }

\newcommand{\squishlisttwo}{
 \begin{list}{$\bullet$}
  { \setlength{\itemsep}{0pt}
     \setlength{\parsep}{0pt}
    \setlength{\topsep}{0pt}
    \setlength{\partopsep}{0pt}
    \setlength{\leftmargin}{2em}
    \setlength{\labelwidth}{1.5em}
    \setlength{\labelsep}{0.5em} } }

\newcommand{\squishend}{
  \end{list}  }

\newcommand\widen{\ensuremath{\mathsf{widen}}}

\newtheorem{theorem}{Theorem}
\newtheorem{corollary}{Corollary}
\newtheorem{lemma}{Lemma}

\begin{document}

\conferenceinfo{WXYZ '05}{date, City.}
\copyrightyear{2005}
\copyrightdata{[to be supplied]}


\title{Higher-Order Symbolic Execution via Contracts}
\authorinfo{Sam Tobin-Hochstadt \and David Van Horn}
           {Northeastern University}
           { \{samth,dvanhorn\}@ccs.neu.edu}



\maketitle
\begin{abstract}

  We present a new approach to automated reasoning about higher-order
  programs by extending symbolic execution to use behavioral contracts
  as symbolic values, enabling \emph{symbolic approximation of
    higher-order behavior}.

  Our approach is based on the idea of an \emph{abstract} reduction
  semantics that gives an operational semantics to programs with both
  concrete and symbolic components.  Symbolic components are
  approximated by their contract and our semantics gives an
  operational interpretation of contracts-as-values.  The result is a
  executable semantics that soundly predicts program behavior,
  including contract failures, for all possible instantiations of
  symbolic components.  We show that our approach scales to an
  expressive language of contracts including arbitrary programs
  embedded as predicates, dependent function contracts, and recursive
  contracts.  Supporting this feature-rich language of specifications
  leads to powerful symbolic reasoning using existing program
  assertions.

  We then apply our approach to produce a verifier for contract
  correctness of components, including a sound and computable
  approximation to our semantics that facilitates fully automated
  contract verification.  Our implementation is capable of verifying
  contracts expressed in existing programs, and of justifying valuable
  contract-elimination optimizations.
\end{abstract}




\section{Behavioral contracts as symbolic values}

Whether in the context of dynamically
loaded JavaScript programs, low-level native C code,
widely-distributed libraries, or simply intractably large code bases,
automated reasoning tools must cope with access to only part of the
program.
To handle missing components, the omitted portions are often assumed
to have arbitrary behavior, greatly limiting the precision and
effectiveness of the tool.

Of course, programmers using external components do not make such
conservative assumptions.  Instead, they attach \emph{specifications}
to these components, often with dynamic enforcement.
  These specifications increase their ability to
reason about programs that are only partially known.
But reasoning solely at the level of specification can also make
verification and analysis challenging as well as requiring substantial effort
to write sufficient specifications.

The problem of program analysis and verification in the presence of
missing \emph{data} has been widely studied, producing many effective
tools that apply \emph{symbolic execution} to non-deterministically
consider many or all possible inputs.  These tools typically determine
constraints on the missing data, and reason using these constraints.
Since the central lesson of higher-order programming is that
computation \emph{is} data, we propose symbolic execution of
higher-order programs for reasoning about systems with omitted
components, taking specifications to be our constraints.

Our approach to higher-order symbolic execution therefore combines
specification-based symbolic reasoning about opaque components with
semantics-based concrete reasoning about available components; we
characterize this technique as \emph{specifications as values}.
As specifications, we adopt higher-order behavioral software
contracts.  Contracts have two crucial advantages for our
strategy. First, they provide benefit to programmers outside of
verification, since they automatically and dynamically enforce their
described invariants.
Because of this, modern languages such as C\#, Haskell and Racket come
with rich contract libraries which programmers already
use~\cite{dvanhorn:Hinze2006Typed,dvanhorn:Fahndrich2010Embedded,dvanhorn:Findler2002Contracts}.
Rather than requiring programmers to annotate code with assertions, we
leverage the large body of code that already attaches contracts at
code boundaries.
For example, the Racket standard library features more than 4000 uses
of contracts~\cite{greenberg}.  Second, the meaning of contracts as
specifications is neatly captured by their dynamic semantics.  As we
shall see, we are able to turn the semantics of contract systems
into tools for verification of programs with contracts.  Verifying
contracts holds promise both for ensuring correctness and improving
performance: in existing Racket code, contract checks take more than
half of the running time for large computations such as rendering 
documentation and type checking large programs~\cite{chaperones}.

Our plan is as follows: we begin with a review of contracts in the
setting of Contract PCF~\cite{dvanhorn:Dimoulas2011Contract} (\S\ref{sec:cpcf}).
Next, we extend Contract PCF with abstract values described by specifications,
producing a core model of symbolic execution for our language of
higher-order contracts, which we dub Symbolic PCF with Contracts (\S\ref{sec:scpcf}).  
This allows us to give
non-deterministic behavior to programs in which any number of
modules are omitted, represented only by their
specifications; here given as contracts.  We accomplish this by
treating contracts as \emph{abstract values}, with the behavior of any
of their possible concrete instantiations.

Contracts as abstract values provides
a rich domain of symbols, including precise specifications for
abstract higher-order values.  These values present new complications
to soundness, addressed with a \emph{demonic context}, a
universal context for discovering blame for behavioral values
(\S\ref{sec:soundness}).

\def\econ{{{\texttt{even?}}}}

\def\smallcon{\texttt{\econ\ -> \econ}}

\def\bigcon{\texttt{(\smallcon) -> (\smallcon)}}

\nocaptionrule

\begin{figure}[t!]
  \noindent
  \begin{minipage}{\columnwidth}
    \hbra
\begin{alltt}
(define-contract list/c
  (rec/c X (or/c empty? (cons/c nat? X))))
(module opaque
  (provide
    [insert (nat? (and/c list/c sorted?)
             -> (and/c list/c sorted?))]
    [nums list/c])))
(module insertion-sort
  (require opaque)
  (define (foldl f l b)
    (if (empty? l) b
        (foldl f (cdr l) (f (car l) b))))
  (define (sort l) (foldl insert l empty))
  (provide
    [sort
     (list/c -> (and/c list/c sorted?))]))
> (sort nums)
(\(\bullet\) (and/c list/c sorted?))
\end{alltt}
\hket
  \end{minipage}
\caption{Verification of insertion sort}
\label{fig:insertion}
\end{figure}

We then extend this core calculus to a model of programs with
modules---including opaque modules whose implementation is not
available---and a much richer contract language
(\S\ref{sec:core-racket}), modeling the functional core of
Racket~\cite{dvanhorn:plt-tr1}.  We show that our symbolic execution
strategy soundly scales up from Symbolic PCF to this more complex
language while preserving its advantages in higher-order reasoning.
Moreover, the technique of describing symbolic values with contracts
becomes even more valuable in an untyped setting.

As the modular semantics is uncomputable, this verification strategy
is necessarily incomplete.  To address this, we apply the technique of
\emph{abstracting abstract
  machines}~\cite{dvanhorn:VanHorn2010Abstracting} to derive first an
abstract machine and then a computable approximation to our semantics
directly from our reduction system (\S\ref{sec:machine}).  We then
turn our semantics into a tool for program verification which is
integrated into the Racket toolchain and IDE
(\S\ref{sec:implementation}).  Users can click a button and explore
the behavior of their program in the presence of opaque modules,
either with a potentially non-terminating semantics, or with a
computable approximation.
Finally, we consider the extensive prior work in symbolic execution,
verification of specifications, and analysis of higher-order programs
(\S\ref{sec:related-work}) and conclude.

Our semantics allows us to use contracts for
verification in two senses: to verify that programs do not violate
their contracts, and verifying rich
properties of programs by expressing them as contracts.
In fact, the semantics alone is, in itself, a program verifier.  The
execution of a modular program which runs without contract errors on
any path is a verification that the concrete portions of the program
never violate their contracts, no matter the instantiation of the
omitted portions.
  This
technique is surprisingly effective, particularly in systems with many
layers, each of which use contracts at their boundaries.  For example,
the implementation of insertion sort in figure~\ref{fig:insertion} is
verified to live up to its contract, which states that it always
produces a sorted list.  This verification works despite the omitted
\texttt{insert} function, used in higher-order fashion as an argument
to \texttt{foldl}.

\paragraph{Contributions}
We make the following contributions:

\begin{enumerate}
\item We propose \emph{abstract reduction semantics} as technique for
  higher-order symbolic execution. This is a variant of operational
  semantics that treats specifications as values, to enable modular
  reasoning about partially unknown programs.
\item We give an abstract semantics for a core typed functional
  language with contracts that equips symbolic values represented as
  sets of contracts with an operational interpretation, allowing sound
  reasoning about opaque program components with rich specifications
  by soundly predicting program behavior for all possible
  instantiations of those opaque components.  We then scale this
  semantics up to model a more realistic untyped language with modules
  and an expressive set of contract combinators.
\item We derive a sound and computable program analysis based on our
  semantics that can serve as the basis for automated program
  verification, optimization, and static debugging.
\item We provide a prototype implementation of an interactive
  verification environment based on our theoretical models which
  successfully verifies existing programs with contracts.
\end{enumerate}

\section{Contracts and Contract PCF}
\label{sec:tools} \label{sec:cpcf}

The basic building block of our specification system is behavioral software
contracts.  Originally introduced by Meyer~\cite{samth:meyer-eiffel}, contracts are
executable specifications that sit at the boundary between software
components.
In a first-order setting, properly
assessing which component violated a contract at run-time is straightforward.  Matters are
complicated when higher-order values such as functions or objects are included
in the language.  Findler and Felleisen
\cite{dvanhorn:Findler2002Contracts} introduced the notion of
\emph{blame} and
 established a semantic framework for properly assessing blame at
run-time in a higher-order language, providing the theoretical basis
for contract systems such as Racket's.

\begin{figure}[h]
  \begin{minipage}{\columnwidth}
    \hbra
\begin{alltt}
(module double
  (provide [dbl ((even? -> even?)
                 -> (even? -> even?))])
  (define dbl (\(\uplambda\) (f) (\(\uplambda\) (x) (f (f x))))))
> (dbl (\(\uplambda\) (x) 7))
\end{alltt}
\hrule
\vspace{3pt}
\emph{top-level broke the contract} 
\emph{on {\tt dbl};}\\
\emph{expected {\tt <even?>}, given: {\tt 7}}

\hket
\end{minipage}
\end{figure}
To illustrate, consider the program above, which consists of a module
and top-level expression.  Module
{\tt double} provides a {\tt dbl} function that implements twice-iterated application,
operating on functions on even numbers.  The top-level
expression makes use of the {\tt dbl} function, but
incorrectly---{\tt dbl} is applied to a function that produces 7.

Contract checking and blame assignment in a higher-order program is
complicated by the fact that is not decidable in general whether the
argument of {\tt dbl} is a function from and to even numbers.  Thus,
higher-order contracts are pushed down into delayed lower-order
checks, but care must be taken to get blame right.
In our example, the top-level is blamed, and rightly so, even though
  {\tt even?}   witnesses the
violation when {\tt f} is applied to {\tt x} while executing
\texttt{dbl}.

\subsection{Contract PCF}

Dimoulas et
al. \cite{samth:dimoulas-POPL2011,dvanhorn:Dimoulas2011Contract}
introduce Contract PCF as a core calculus for the investigation of
contracts, which we take as the starting point for our model.  CPCF
extends PCF \cite{dvanhorn:Plotkin1977LCF} with contracts for base
values and first-class functions.  We provide a brief recap of the
syntax and semantics of CPCF.

\begin{figure}[t!]
\begin{displayfig}[]{PCF with Contracts}
\[
\begin{array}{ll@{\;}c@{\;}l}
\text{Types} &
\mtyp & ::= & \mbasetyp\ |\ \starr\mtyp\mtyp\ |\ \stcon\mtyp \\
\text{Base types} &
\mbasetyp & ::= & \stypnum\ |\ \stypbool
\\[1mm]
\text{Terms} &
\mexp & ::= & \mans\ |\ \mvar\ |\ \sapp\mexp\mexp\
|\ \strec\mvar\mtyp\mexp\ |\ \sif\mexp\mexp\mexp\\
&&|&\sopone{\mop_1}\mexp\ |\ \soptwo\mexp{\mop_2}\mexp\ |\ \chk\mcon\mmodvar\mmodvar\mmodvar\mexp
\\[1mm]
\text{Operations} &
\mop_1 &::=& \syntax{zero?}\ |\ \syntax{false?}\ |\ \dots\\
&\mop_2 &::=& +\ |\ -\ |\ \wedge\ |\ \vee\ |\ \dots
\\[1mm]
\text{Contracts} &
\mcon & ::= & \sflat\mexp\ |\ \sarr\mcon\mcon\ |\ \stdep\mcon\mvar\mtyp\mcon
\\[1mm]
\text{Answers} &
\mans & ::= & \mval\ |\ \mctx[\sblm\mmodvar\mmodvar]
\\[1mm]
\text{Values} &
\mval & ::= & \stlam\mvar\mtyp\mexp\ |\ 0\ |\ 1\ |\ -1\ |\ \dots\ |\ \strue\ |\ \sfalse
\\[1mm]
\text{Evaluation} &
\mctx & ::= & [\;]\ |\ \sapp\mctx\mexp\ |\ \sapp\mval\mctx\ |\ \soptwo\mctx{\mop_2}\mexp\
|\ \soptwo\mval{\mop_2}\mctx\\
\text{contexts} &&|& \sopone{\mop_1}\mctx\ |\ \sif\mctx\mexp\mexp\ |\ \chk\mcon\mmodvar\mmodvaro\mmodvaroo\mctx
\end{array}
\]
\end{displayfig}

\begin{displayfig}[$\mexp\stdstep\mexpo$]{Semantics for PCF with Contracts}
\[
\begin{array}{@{\;}r@{\ \ }c@{\ \ }l@{\;}}
\sif\strue{\mexp_1}{\mexp_2} &\stdstep& \mexp_1
\\[1mm]
\sif\sfalse{\mexp_1}{\mexp_2} &\stdstep& \mexp_2
\\[1mm]
\sapp{(\stlam\mvar\mtyp\mexp)}\mval &\stdstep&
\subst\mval\mvar\mexp
\\[1mm]
\strec\mvar\mtyp\mexp &\stdstep& \subst{\strec\mvar\mtyp\mexp}\mvar\mexp
\\[1mm]
\sopone\mop{\vec\mval}
&\stdstep&
\mans \mbox{ if } \delta(\mop,\vec\mval) = \mans
\\[1mm]
\chk{\stdep{\mcon_1}\mvar\mtyp{\mcon_2}}\mmodvar\mmodvaro\mmodvaroo\mval
& \stdstep &
\\
\multicolumn{3}{r@{\;}}{
\stlam\mvar\mtyp{\chk{\mcon_2}\mmodvar\mmodvaro\mmodvaroo{(\sapp\mval{\chk{\mcon_1}\mmodvaro\mmodvar\mmodvaroo\mvar})}}}
\\[1mm]
\chk{\sflat\mexp}\mmodvar\mmodvaro\mmodvaroo\mval
& \stdstep &
\sif{(\sapp\mexp\mval)}\mval{\sblm\mmodvar\mmodvaroo}
\end{array}
\]
\end{displayfig}

\end{figure}

Contracts for flat values, $\sflat\mexp$, employ predicates that may
use the full expressive power of CPCF.  Function contracts,
$\sarr{\mcon_1}{\mcon_2}$ consist of a pre-condition contract
$\mcon_1$ for the argument to the function and a post-condition
contract $\mcon_2$ for the function's result.  Dependent function
contracts, $\sdep{\mcon_1}\mvar{\mcon_2}$, bind $\mvar$ to the
argument of the function in the post-condition contract $\mcon_2$, and
thus express an dependency between a functions input and result.  In
the remainder of the paper, we treat the non-dependent function
contract $\sarr{\mcon_1}{\mcon_2}$ as
shorthand for $\sdep{\mcon_1}\mvar{\mcon_2}$ where $\mvar$ is fresh.

A contract $\mcon$ is attached to an expression with the monitor
construct $\chk\mcon\mmodvar\mmodvaro\mmodvaroo\mexp$, which carries
three labels: $\mmodvar$, $\mmodvaro$, and $\mmodvaroo$, denoting the
names of components.  (An implementation would synthesize these names
from the source code.)  The monitor checks any interaction between the
expression and its context is in accordance with the contract.

Component labels play an important role in case a contract failure is
detected during contract checking.  In such a case, blame is assigned
with the $\sblm\mmodvar\mmodvaro$ construct, which denotes the
component named $\mmodvar$ broke its contract with $\mmodvaro$.

CPCF is equipped with a standard type system for PCF plus the addition
of a contract type $\stcon\mtyp$, which denotes the set of contracts
for values of type
$\mtyp$~\cite{samth:dimoulas-POPL2011,dvanhorn:Dimoulas2011Contract}.
The type system is straightforward, so for the sake of space, we defer
the details to an appendix (\S\ref{sec:cpcf-types}).

The semantics of CPCF are given as a call-by-value reduction relation
on programs.  One-step reduction is written as $\mexp \stdstep \mexpo$
and defined as the above relation, contextually closed over evaluation
contexts $\mctx$.  The reflexive transitive closure of one-step
reduction is written $\mexp \multistdstep \mexpo$.

The first five cases of the reduction relation are standard for PCF.
The remaining two cases implement contract checking for function
contracts and flat contracts, respectively.
The monitor of a function contract on a function reduces to a function
that monitors its input with reversed blame labels and monitors its
output with the original blame labels.%
\footnote{For brevity, we have presented the so-called \emph{lax}
  dependent contract rule, although our implementation uses
  \emph{indy}~\cite{samth:dimoulas-POPL2011}, which is obtained by
  replacing the right-hand side with:
\[
\stlam\mvar\mtyp{\chk{\subst{\chk{\mcon_1}\mmodvar\mmodvaroo\mmodvaroo\mvar}\mvar{\mcon_2}}\mmodvar\mmodvaro\mmodvaroo{\sapp\mval{\chk{\mcon_1}\mmodvaro\mmodvar\mmodvaroo\mvar}}}\text.
\]
}
The monitor of a flat contract reduces to an \syntax{if}-expression
which tests whether the predicate holds.  If it does, the value is
returned.  If it doesn't, a contract error is signaled with the
appropriate blame.

\section{Symbolic PCF with Contracts}
\label{sec:scpcf}

We now describe an extension to Contract PCF that enriches the
language with \emph{symbolic values}, drawn from the language of
contracts, and show the revised semantics.  The basic idea of SCPCF is
to take the values of CPCF as ``pre''-values $\mpreval$ and add a
notion of an unknown values (of type $\mtyp$), written ``$\topaque$''.
Purely unknown values have arbitrary behavior, but we will refine
unknowns by attaching a set of contracts that specify an agreement
between the program and the missing component.  Such refinements can
guide an operational characterization of a program. Pre-values are
refined by a set of contracts to form a value,
$\with\mpreval\mconset$, where $\mconset$ ranges over sets of
contracts.

The high-level goal of the following semantics is to enable the
running of programs with unknown components.  The main requirement is
that the results of running such computations should soundly and
precisely approximate the result of running that same program after
replacing an unknown with \emph{any} allowable value.  More precisely,
if a program involving some value $\mval$ produces an answer $\mans$,
then abstracting away that value to an unknown should produce an
approximation of $\mans$:
\[
\text{if }\mctx[\mval] \multistdstep \mans\text{ and }\vdash\mval:\mtyp\text{, then }\mctx[\topaque] \multistdstep \mans'\text,
\]
where $\mans'$ ``approximates'' $\mans$ in way that is made
formal in section \ref{sec:soundness}.

\begin{figure}
\begin{displayfig}[]{Symbolic PCF with Contracts}
\[
\begin{array}{ll@{\;}c@{\;}l}
\text{Prevalues} &
\mpreval & ::= &
\topaque\ |\ \stlam\mvar\mtyp\mexp\ |\ 0\ |\ 1\ |\ -1\ |\ \!\dots\!\ |\ \strue\ |\ \sfalse
\\[1mm]
\text{Values} &
\mval & ::= &
\with\mpreval{\{\mcon,\dots\}}
\end{array}
\]
\end{displayfig}

\begin{displayfig}[$\mexp\stdstep\mexp'$]{Semantics for Symbolic PCF with Contracts}
\[
\begin{array}{@{\;}r@{\ \ }c@{\ \ }l@{\;}}
\sif\mval{\mexp_1}{\mexp_2} &\stdstep& \mexp_1
\mbox{ if }\absdelta(\struep,\mval) \ni \strue
\\[1mm]
\sif\mval{\mexp_1}{\mexp_2} &\stdstep& \mexp_2
\mbox{ if }\absdelta(\struep,\mval) \ni \sfalse
\\[1mm]
\sapp{(\stlam\mvar\mtyp\mexp)}\mval &\stdstep&
\subst\mval\mvar\mexp
\\[1mm]
\strec\mvar\mtyp\mexp &\stdstep& \subst{\strec\mvar\mtyp\mexp}\mvar\mexp
\\[1mm]
\sopone\mop{\vec\mval}
&\stdstep&
\mans \mbox{ if } \absdelta(\mop,\vec\mval) \ni \mans
\\[1mm]
\sapp{(\with{\opaque^{\starr\mtyp{\mtyp'}}\!}\mconset)}{\mval}
&\stdstep &
\with{\opaque^{\mtyp'}}{\{\subst\mval\mvar{\mcon_2}\ |\ \stdep{\mcon_1}\mvar\mtyp{\mcon_2} \in \mconset\}}
\\[1mm]
\sapp{(\with{\opaque^{\starr\mtyp{\mtyp'}}\!}\mconset)}{\mval}
&\stdstep &
\sapp{\syntax{havoc}_\mtyp}\mval
\end{array}
\]
\end{displayfig}
\end{figure}

\paragraph{Notation:} Abstract (or synonymously: symbolic) values $\ab\mval$ range over values of the form
$\with\topaque\mconset$.
Whenever the refinement of a value is irrelevant, we omit the
$\mconset$ set.  We write $\mval\cdot\mcon$ for $\with\mpreval{\mconset
  \cup \{\mcon\}}$ where $\mval = \with\mpreval\mconset$.

The semantics given above replace that of section \ref{sec:cpcf},
equipping the operational semantics with an interpretation of symbolic
values.  (The semantics of contract checking is deferred for the
moment.)

To do so requires two changes:
\begin{enumerate}
\item the $\delta$ relation must be extended to interpret operations
  when applied to symbolic values, and
\item the one-step reduction relation must be extended to the case of
 (1) branching on a
  (potentially) symbolic value, and
 (2) applying a symbolic function.
\end{enumerate}

\subsection{Operations on symbolic values}

Typically, the interpretation of operations is defined by a function
$\delta$ that maps an operation and argument values to an answer.  So
for example, $\delta(\ssucc,0)=1$.  The result of
applying a primitive may either be a value in case the operation is
defined on its given arguments, or blame in case it is not.

The extension of $\delta$ to interpret symbolic values is largely
straightforward.  It starts by generalizing $\delta$ from a function
from an operation and values to an answer, to a relation between
operations, values, and answers (or equivalently, to a function from
an operation and values to \emph{sets} of answers). This enables
multiple results when a symbolic value does not convey enough
information to uniquely determine a single result.  For example,
$\delta(\szero,0) = \{\strue\}$, but $\delta(\szero,\opaque^\stypnum)
= \{\strue,\sfalse\}$.  From here, all that remains is adding
appropriate clauses to the definition of $\delta$ for handling
symbolic values.  As an example, the definition includes:
\[
\delta(\splus,\mval_1,\mval_2) \ni \opaque^\stypnum
\text{, if }\mval_1\text{ or }\mval_2 = \with{\opaque^\stypnum}\mconset\text.
\]
The remaining cases are similarly straightforward.

The revised reduction relation reduces an operation,
non-deterministically, to any answer in the $\delta$ relation.

%

\subsection{Branching on symbolic values}

The shift from the semantics of section~\ref{sec:cpcf} to
section~\ref{sec:scpcf} involves what appears to be a cosmetic change
in the reduction of conditionals, e.g., from
\[
\sif\strue{\mexp_1}{\mexp_2} \stdstep \mexp_1
\]
to
\[
\sif\mval{\mexp_1}{\mexp_2} \stdstep \mexp_1
\mbox{ if }\absdelta(\struep,\mval) \ni \strue\text.
\]
In the absence of symbolic values, the two relations are equivalent,
but once symbolic values are introduced, the latter handles branching
on potentially symbolic values by deferring to $\delta$ to determine
if $\mval$ is possibly true.  Consequently branching on
$\opaque^\stypbool$ results in both $\mexp_1$ and $\mexp_2$ since
$\delta(\struep,\opaque^\stypbool) = \{\strue,\sfalse\}$.  Without
this slight refactoring for conditionals, additional cases for the
reduction relation are required, and these cases would largely mimic
the existing $\syntax{if}$ reductions.  By reformulating in terms of
$\delta$, we enable the uniform reduction of abstract and concrete
values.

\subsection{Applying symbolic functions}

When applying a symbolic function, the reduction relation must take
two distinct possibilities into account.
The first is that the argument to the symbolic function escapes, but
no failure occurs in the unknown context, so the function returns an
abstract value refined by the range contracts of the function.
The second is that the use of $\mval$ in a unknown context results in
the blame of $\mval$.  To discover if blaming $\mval$ is possible we
rely upon a \syntax{havoc} function, which iteratively explores the
behavior of $\mval$ for possible blame.  Its only purpose is to
uncover blame, thus it never produces a value; it either diverges or
blames $\mval$.
In this simplified model, the only behavioral values are functions, so
we represent all possible uses of the escaped value by iteratively
applying it to unknown values.  This construction represents a
universal ``demonic'' context to discover a way to blame $\mval$ if
possible, and we have named the function \syntax{havoc} to emphasize
the analogy to Boogie's \syntax{havoc}
function~\cite{dvanhorn:Fahndrich2011Static}, which serves the same
purpose, but in a first-order setting.

\def\hv{\syntax{havoc}}
\def\varx{\syntax{x}}

The \hv function is indexed by the type of its argument.
At base type, values do not have behavior, so \hv simply
produces a diverging computation.  At function type, \hv
produces a function that applies its argument to an appropriately
typed unknown input value and then recursively applies \hv
at the result type to the output:
\begin{align*}
\syntax{havoc}_\mbasetyp &= \srec\varx\varx\\
\syntax{havoc}_{\starr\mtyp{\mtyp'}} &= \stlam\varx{\starr\mtyp{\mtyp'}}{\syntax{havoc}_{\mtyp'}(\sapp\varx\topaque)}
\end{align*}

To see how \hv finds all possible errors in a term, consider the
following function guarded by a contract:
\[
\achksimple{\stlam{\syntax{f}}{\starr\stypnum\stypnum}{\sapp{\syntax{sqrt}}{(\sapp{\syntax{f}}0)}}}{\sarr{(\sarr\sany\sany)}\sany}
\]
 where \syntax{any} is the trivial contract
 $\liftpred{\slamnp{\syntax{x}}\strue}$ and \syntax{sqrt}
has the type $\starr\stypnum\stypnum$ and contract
$\sarr{\liftpred{\syntax{positive?}}}{\liftpred{\syntax{positive?}}}$.
If we then apply \hv to this term at the appropriate type, it will
supply the input $\opaque^{\starr\stypnum\stypnum}$ for \syntax{f}.
When this abstract value
is applied to $0$, it reduces to both a diverging term that produces
no blame, and the symbolic number $\opaque^\stypnum$.  Finally,
\syntax{sqrt} is applied to $\opaque^\stypnum$, which both passes and
fails the contract check on \syntax{sqrt}, since $\opaque^\stypnum$
represents both positive and non-positive numbers; the latter demonstrates
the original function could be blamed.

In contrast, if the original term was wrapped in the contract
$\sarr{(\sarr\sany{\liftpred{\syntax{positive?}}})}\sany$, then the abstract value
$\opaque^{\starr\stypnum\stypnum}$ would have been wrapped in the
contract $\sarr\sany{\liftpred{\syntax{positive?}}}$. When the
wrapped abstract function is applied to $0$, it then produces the
more precise
abstract value
$\opaque^\stypnum\cdot\liftpred{\syntax{positive?}}$ as the input to
\syntax{sqrt} and fails to blame the original function.

The ability of \hv to find blame if possible is key to our soundness
result.

\subsection{Contract checking symbolic values}
\label{sec:checking-scpcf}

We now turn to the revised semantics for contract checking reductions
in the presence of symbolic values.  The key ideas are that we
\begin{enumerate}
\item avoid checking any contracts which a value provably
  satisfies, and
\item add flat contracts to a value's refinement set whenever
a contract check against that value succeeds.
\end{enumerate}

To implement the first idea, we add a reduction relation which
sidesteps a contract check and just produces the checked value
whenever the value proves it satisfies the contract.  To implement the
second idea, we revise the flat contract checking reduction relation
to produce not just the value, but the value refined by the contract
in the success branch of a flat contract check.

\begin{display}{Contract checking, revisited}
\[
\begin{array}{@{}r@{\;\;}c@{\;\;}llr@{}}
\chk\mcon\mmodvar\mmodvaro\mmodvaroo\mval
&\stdstep&
\mval \mbox{ if } \proves\mval\mcon
\\[1mm]
\chk{\liftpred\mexp\relax}\mmodvar\mmodvaro\mmodvaroo\mval
&\stdstep
\\
\multicolumn{3}{r@{\;}}{
\sif{(\sapp\mexp\mval)}{(\mval\cdot{\liftpred\mexp})}{\simpleblm\mmodvar\mmodvaro}
\mbox{ if }\not\proves\mval{\liftpred\mexp}}
\\[1mm]
\chk{\stdep{\mcon_1}\mvar\mtyp{\mcon_2}}\mmodvar\mmodvaro\mmodvaroo\mval
& \stdstep &
\\
\multicolumn{3}{r@{\;}}{
\stlam\mvar\mtyp{\chk{\mcon_2}\mmodvar\mmodvaro\mmodvaroo{\sapp\mval{\chk{\mcon_1}\mmodvaro\mmodvar\mmodvaroo\mvar}}}}
\end{array}
\]
\end{display}

The judgment $\proves\mval\mcon$ denotes that $\mval$ provably
satisfies the contract $\mcon$, which we read as ``$\mval$ proves
$\mcon$.''  Our system is parametric with respect to this provability
relation and the precision of the symbolic semantics improves as the
power of the proof system increases.  For concreteness, we consider
the following simple, yet useful proof system which asserts a symbolic
value proves any contract it is refined by:
\[
\inferrule{\mcon\in\mconset}
          {\proves{\with\mval\mconset}\mcon}
\]
As we will see subsequently, this relation can easily be extended to
handle more sophisticated reasoning.

Taken together, the revised contract checking relation and proves
relation allow values to remember contracts once they are checked and
to avoid rechecking in subsequent computations.  Consider the
following program with abstract pieces:
$$
\begin{tabular}{c@{\;}l@{\,}c@{\,}l}
\syntax{let}&\syntax{keygen}&\syntax{=}&\achksimple{\sarr{\syntax{unit}}{\liftpred{\syntax{prime?}}}}{\opaque}\\
 &\syntax{rsa}&\syntax{=}&\achksimple{\sarr{\liftpred{\syntax{prime?}}}{(\sarr\sany\sany})}{\opaque}\\
\multicolumn{4}{l}{\syntax{in}\ \sapp{\sapp{\syntax{rsa}}{(\sapp{\syntax{keygen}}{\syntax{()}})}}{\syntax{``Plaintext"}}}
\end{tabular}
$$
When invoking \syntax{keygen} produces an abstract number, it will be checked
against the \syntax{prime?} contract, which will non-deterministically
both succeed and fail, since \syntax{keygen}'s source is not available
to be verified.  However, in the case where the check succeeds, the
\syntax{prime?} contract is \emph{remembered}, meaning that our
semantics correctly predicts that the top level application does not
break \syntax{rsa}'s contract by providing a composite number.  This
verifies that regardless of the implementation of \syntax{keygen} and
\syntax{rsa}, which may themselves be buggy, their \emph{composition} is
verified to uphold its obligations.

\subsection{Soundness}
\label{sec:soundness}

Soundness relies on the definition of approximation between terms.  We
write $\mexp \sqsubseteq \mexpo$ to mean $\mexpo$ \emph{approximates}
$\mexp$, or conversely $\mexp$ \emph{refines} $\mexpo$. The basic
intuition for approximation is an abstract value, which can be thought
of as standing for a set of acceptable concrete values, approximates a
concrete value if that value is in the set the abstract value denotes.

Since the $\topaque$ value stands for any value of type $\mtyp$, we
have the following axiom:
\[
\inferrule{\vdash \mval : \mtyp }
          {\mval\sqsubseteq\topaque}
\]
(In subsequent judgments, we assume both sides of the approximation
relation are typable at the same type and thus omit type annotations
and judgments.) A monitored expression is approximated by its
contract:
\[
\inferrule{ }
          {\achksimple\mcon\mexp \sqsubseteq \opaque\cdot\mcon}
\]
To handle the approximation of wrapped functions, we employ the
following rule, which matches the right-hand side of the reduction
relation for the monitor of a function:
\[
\inferrule{ }
          {\slam\mvar{\achksimple\mcono{(\sapp\mval{\achksimple\mcon\mvar})}}
            \sqsubseteq
            \opaque\cdot\sdep\mcon\mvar\mcono}
\]
Arbitrary contract refinements may be introduced as follows:
\begin{mathpar}
\inferrule{ }
          {\mval\cdot\mcon\sqsubseteq\mval}

\inferrule{\mval\sqsubseteq\mval'}
          {\mval\cdot\mcon\sqsubseteq\mval'\cdot\mcon}
\end{mathpar}
A contract may be eliminated when a value proves it:
\[
  \inferrule{\proves\mval\mcon}
            {\mval\sqsubseteq\mval\cdot\mcon}
\]
If an expression approximates a monitored expression, it's OK to
monitor the approximating expression too:
\[
  \inferrule
      {\achksimple\mcon\mexp \sqsubseteq \mexpo}
      {\achksimple\mcon\mexp \sqsubseteq \achksimple\mcon\mexpo}
\]
Finally, $\sqsubseteq$ is reflexively, transitively, and compatibly
closed.

\begin{theorem}[Soundness of Symbolic PCF with Contracts]
\label{thm:soundness-scpcf}
If $\mexp \sqsubseteq \mexp'$ and $\mexp \multistdstep \mans$, then
there exists some $\mans'$ such that $\mexp' \multistdstep \mans'$
where $\mans \sqsubseteq \mans'$.
\end{theorem}
\begin{proof}(Sketch)
The proof follows from
(1) a completeness result for $\syntax{havoc}$,
which states that if $\mctx[\mval] \multistdstep
\mctx'[\simpleblm\mlab\relax]$ where $\mlab$ is not in $\mctx$, then
$\sapp{\syntax{havoc}}\mval\multistdstep\mctx''[\simpleblm\mlab\relax]$,
and
(2) the following main lemma:
if $\mexp_1\stdstep\mexpo_1 \neq\mctx[\simpleblm\mmodvar\mmodvaro]$,
and $\mexp_1 \sqsubseteq \mexp_2$, then $\mexp_2
\multistdstep\mexpo_2$ and $\mexpo_1 \sqsubseteq \mexpo_2$, which is
in turn proved reasoning by cases on $\mexp_1\stdstep\mexpo_1$ and
appealing to auxiliary lemmas that show approximation is preserved by
substitution and primitive operations.
The full proof for the enriched system of
section~\ref{sec:core-racket} is given in appendix~\ref{sec:proofs}.
\end{proof}

The soundness result achieves the high-level goal stated at the
beginning of this section: we have constructed an \emph{abstract
  reduction semantics} for the sound symbolic execution of programs
such that their symbolic execution approximates the behavior of
programs for \emph{all possible instantiations} of the opaque
components.  In particular, we can verify pieces of programs by
running them with missing components, refined by contracts.  If the
abstract program does not blame the known components, \emph{no context
  can cause those components to be blamed.}

\section{Symbolic Core Racket}
\label{sec:core-racket}

Having developed the core ideas of our symbolic executor for programs
with contracts, we extend our language to an \emph{untyped} core calculus of
modular programs with data structures and rich contracts.  This forms
a core model of a realistic programming language,
Racket~\cite{dvanhorn:plt-tr1}.  In addition to closely modeling our
target language, omitting types places a greater burden on the
contract system and symbolic executor.  As we see in this section,
ours is up to the job.

To SCPCF we add pairs, the empty list, and related operations;
contracts on pairs; recursive contracts; and conjunctive and
disjunctive contracts.  Predicates, as before, are expressed as
arbitrary programs within the language itself.  Programs are organized
as a set of module definitions, which associate a module name with a
value and a contract.  Contracts are established at module boundaries
and here express an agreement between a module and the external
context.  The contract checking portion of the reduction semantics
monitors these agreements, maintaining sufficient information to blame
the appropriate party in case a contract is broken.

\subsection{Syntax}

The syntax of our language is given in figure~\ref{fig:stx-rkt}. We write
$\vec{\mexp}$ for a possibly-empty sequence of $\mexp$, and treat
these sequences as sets where convenient. Portions highlighted in
\graybox{\mbox{gray}} are the key extensions over SCPCF, as presented
in earlier sections.

\begin{figure}[t]
\begin{displayfig}[]{}
\[
\begin{array}{@{\,}l@{\;}cl}
\mprg,\mprgo & ::= & \graybox{\mvmod \mexp}\\
\mmod,\mmodo & ::= & \graybox{\smod\mmodvar\mcon\mval}\\
\mexp,\mexpo & ::= & \graybox{\mmodvar^\mlab}\ |\ \mvar\ |\ \mans\
           |\ \sapp\mexp\mexp^\mlab\
           |\ \sif{\mexp}{\mexp}{\mexp}\
           |\ \sapp\mop{\vec\mexp}^\mlab\
           |\ \srec\mvar\mexp\
\\
&|& \chk\mcon\mmodvar\mmodvar\mmodvar\mexp\\
\mpreval & ::= & \mnum\
            |\ \strue\ |\ \sfalse\
            |\ \slam\mvar\mexp\
            |\ \opaque\
            |\ \graybox{\vcons\mval\mval\
            |\ \sempty}
\\
\mval & ::= & \with\mpreval\mconset
\\
\mcon,\mcono & ::= & \graybox\mvar\ 
|\ \sdep\mcon\mvar\mcon\ |\ \spred\mexp
\\
& |&   \graybox{ \sconsc\mcon\mcon\ |\ \sorc\mcon\mcon\ |\  \sandc\mcon\mcon
\ |\ \srecc\mvar\mcon}\\
\mop & ::= & \ssucc\ |\ \graybox{\scar\ |\ \scdr\ |\ \sconsop}\
|\ \splus \ |\ \sequalp\  |\ \moppred\ |\ \dots \\
\moppred & ::= & \graybox{\snatp\ |\ \sboolp\ |\ \semptyp\ |\ \sconsp\ |\ \sprocp\ |\
\sfalsep}\ \\
\mans & ::= & \mval\ |\ \mctx[{\ablm\mlab\mlab\mval\mcon\mval}]
\end{array}
\]
\end{displayfig}
\caption{Syntax of Symbolic Core Racket}
\label{fig:stx-rkt}
\end{figure}

A program $\mprg$ consists of a sequence of modules followed by a main
expression.  Modules are second-class entities that name a single value along with a contract
to be applied to that value.  \emph{Opaque modules} are modules whose
body is $\opaque$.  Expressions now include module references,
labeled by the module they appear in; this label is used as the
negative party for the module's contract.  Applications are also labeled;
this label is used if the application fails. Pair values and the empty
list constant are standard, along with their operations.  Since the
language is untyped, we add standard type predicates such as $\snatp$.

The new contract forms include pair contracts, with the obvious
semantics, conjunction and disjunction of contracts, and
recursive contracts with contract variables.

Contract checks $\chk\mcon\mmodvar\mmodvaro\mmodvaroo\mexp$, which
will now be inserted automatically by the operational semantics, take
all of their labels from the names of modules, with the third label
$\mmodvaroo$ represents the module in which the contract originally
appeared. As before, $\mmodvar$ represents the positive party to the
contract, blamed if the expression does not meet the contract, and
$\mmodvaro$ is the negative party, blamed if the context does not
satisfy its obligations.
Whenever these annotations can be inferred from context, we omit them; in
particular, in the definition of relations, it is assumed all
 checks of the form $\achksimple\mcon\mexp$ have the
same annotations.
We omit labels on applications whenever they provably cannot be
blamed, e.g. when the operand is known to be a function.

A blame expression, $\ablm\mlab\mlabo\mval\mcon\mval$, now indicates that
the module (or the top-level expression) named by $\mlab$ broke its
contract with $\mlabo$, which may be a module or the language,
indicated by $\Lambda$ in the case of primitive errors.

\paragraph{Syntactic requirements:}
We make the following assumptions of initial, well-formed programs,
$\mprg$: programs are closed, every module reference and application
is labeled with the enclosing module's name, or $\dagger$ if in the
top-level expression, operations are applied with the correct arity,
abstract values only appear in opaque module definitions, and no
monitors or blame expressions appear in the source program.

We also require that  recursive contracts be \emph{productive}, meaning either a
function or pair contract constructor must occur between binding and
reference.
We also require that contracts  in the source program are
closed, both with respect to $\uplambda$-bound and contract variables.
 Following standard practice, we will say that
a contract is \emph{higher-order} if it
syntactically contains a function contract;
otherwise, the contract is \emph{flat}.  Flat contracts can be checked
immediately, whereas higher-order contracts potentially require
delayed checks.  All predicate contracts are
necessarily flat.

\paragraph{Disjunction of contracts:}
For disjunctions, we require that at most one of the disjuncts is
higher-order and without loss of generality, we assume it is the right
disjunct.  The reason for this restriction is that we must choose at
the time of the \emph{initial} check of the contract which disjunct to
use---we cannot just try both because higher-order checks must be
delayed.  In Racket, disjunction is therefore restricted to contracts
that are distinguishable in a first-order way, which we simplify to
the restriction that only one can be higher-order.

\subsection{Reductions}

Evaluation is modeled with one-step reduction on programs, $\mprg
\stdstep \mprgo$.  Since the module context consists solely of
syntactic values, all computation occurs by reduction of the top-level
expression.  Thus program steps are defined in terms of top-level
expression steps, carried out in the context of several module
definitions.  We model this with a reduction relation on expressions
in a module context, which we write $\mvmod \vdash \mexp \stdstep
\mexpo$.  We omit the the module context where it is not used and
write $\mexp \stdstep \mexpo$ instead.  Our reduction system is given
with evaluation contexts, which are identical to those of SCPCF in
section~\ref{sec:scpcf}.

We present the definition of this relation in several parts.

\subsubsection{Applications, operations, and conditionals}

First, the definition of procedure applications, conditionals,
primitive operations is as usual for a call-by-value language.
Primitive operations are interpreted by a $\delta$ relation (rather
than a function), just as in section~\ref{sec:scpcf}.  The reduction relation
for these terms is defined as follows:
\begin{display}[$\mexp\;\stdstep\;\mexpo$]{Basic reductions}
\[
\begin{array}{rclr}
(\sapp{\slam\mvar\mexp}\mval)^\mlab
&\stdstep&
\subst\mval\mvar\mexp
\\ 
(\sapp{\mval}{\mval'})^\mlab
&\stdstep&
\asblm{\mlab}{\Lambda}{\mval}
&
\mbox{if }\absdeltamap\sprocp\mval\sfalse
\\ 
(\sapp{\mop}{\vec{\mval}})^\mlab
&\stdstep&
\mans
&
\mbox{if }\deltamap{\mop^\mlab}{\vec{\mval}}\mans
\\ 
\sif{\mval}{\mexp}{\mexpo}
&\stdstep&
\mexp
&
\mbox{if }\deltamap\sfalsep\mval\sfalse
\\ 
\sif\mval\mexp\mexpo
&\stdstep&
\mexpo
&
\mbox{if }\deltamap\sfalsep\mval\strue
\\[1mm]
\end{array}
\]
\end{display}
\noindent
Again, we rely on $\delta$ not only to interpret operations, but
also to determine if a value is a procedure or $\sfalse$; this
 allows uniform handling of abstract values, which may (depending on
 their remembered contracts) be treated as both true and false.  We add a reduction to
$\simpleblm\mmodvar\Lambda$ when applications are misused; the program
has here broken the contract with the language, which is no longer
checked statically by the type system as it was in SCPCF.
Additionally, our rules for \syntax{if}
 follow the Lisp tradition (which Racket adopts) in treating all non-$\sfalse$ values as
 true.

\subsubsection{Basic operations}

Basic operations, as with procedures and conditionals, follow SCPCF
closely. Operations on concrete values are standard, and we present
only a few selected cases.  Operations on abstract values are more
interesting.  A few selected cases are given in figure~\ref{fig:delta}
as examples.  Otherwise, the definition of $\delta$ is for concrete
values is standard and we relegate the remainder to
appendix~\ref{sec:absdelta}.

When applying base operations to abstract values, the results are
potentially complex.  For example, $\sapp\ssucc\opaque$ might produce
any natural number, or it might go wrong, depending on what value
$\opaque$ represents.  We represent this in
$\delta$ with a combination of
non-determinism, where $\absdelta$ relates an operation and its inputs
to multiple answers, as well as abstract values as results, to handle
the arbitrary natural numbers or booleans that might be produced.  A
representative selection of the $\absdelta$ definition for abstract
values is presented in figure~\ref{fig:delta}.

\begin{figure}
\begin{displayfig}[$\absdelta(\mop^\mlab, \vec\mval) \ni
  \mans$]{Primitive operations (concrete values)}
\[
\begin{array}{lcl}
&&\absdelta(\ssucc,\mnum) \ni \mnum+1 \\
&&\absdelta(\splus,\mnum,\mnumo)\ni \mnum+\mnumo \\
&&\absdelta(\scar,\vcons\mval{\mval'})\ni\mval \\
&&\absdelta(\scdr,\vcons\mval{\mval'})\ni\mval'
\end{array}
\]
\end{displayfig}
\begin{displayfig}{Primitive operations (abstract values)}
\[
\begin{array}{lcl}
\proves\mval\moppred &\implies& \absdeltamap\moppred\mval\strue\\
\refutes\mval\moppred &\implies& \absdeltamap\moppred\mval\sfalse\\
\ambig\mval\moppred
&\implies& \absdeltamap\moppred\mval{\with\opaque{\{\liftpred\sboolp\}}}\\
\proves\mval\snatp  &\implies& \absdeltamap\ssucc\mval{\with\opaque{\{\liftpred\snatp}\}}\\
\refutes\mval\snatp &\implies& \absdeltamap{\ssucc^\mlab}\mval{\ablm\mlab\ssucc\mval\lambda\mval}\\
\ambig\mval\snatp &\implies& \absdeltamap\ssucc\mval{\with\opaque{\liftpred\snatp}}\\
                  &&\wedge\ \absdeltamap{\ssucc^\mlab}\mval{\ablm\mlab\ssucc\mval\lambda\mval}\\
\proves\mval\sconsp &\implies& \absdeltamap\scar\mval{\projleft(\mval)}\\
\refutes\mval\sconsp &\implies& \absdeltamap{\scar^\mlab}\mval{\ablm\mlab\scar\mval\lambda\mval}
\\
\ambig\mval\sconsp &\implies& \absdeltamap\scar\mval{\projleft(\mval)} \\
&&\wedge\ \absdeltamap{\scar^\mlab}\mval{\ablm\mlab\scar\mval\lambda\mval}
 \end{array}
 \]
\end{displayfig}
\caption{Basic operations}
\label{fig:delta}
\end{figure}

The definition of $\absdelta$ relies on a proof system relating
predicates and values, just as with contract checking.  Here, $\proves\mval\moppred$ means that
$\mval$ is known to satisfy $\moppred$, $\refutes\mval\moppred$
means that
$\mval$ is known not to satisfy $\moppred$, and $\ambig\mval\moppred$
 means $\mval$ neither is known.  For example, $\proves7\snatp$, $\refutes\strue\sconsp$, and
$\ambig{\opaque}\moppred$ for any $\moppred$.  (Again, our system is
parametric with respect to this proof system, although we present a useful
instance in section~\ref{sec:proof-system}.)

Finally, if no case matches, then an error is produced:
\[
\absdelta(\mop^\mlab,\vec\mval) \ni \simpleblm\mlab\Lambda
\]
Labels on operations come from the application site of the operation
in the program, e.g.~$\sapp{\ssucc}5^\mlab$ so that the appropriate
module can be blamed when primitive operations are misused, as in the
last case, and are
omitted whenever they are irrelevant.  When primitive operations are
misused, the violated contract is on $\Lambda$, standing for the
programming language itself, just as in the rule for application of
non-functions.

\subsubsection{Module references}

To handle references to module-bound variables, we define a module
environment that describes the module context $\mvmod$.  Using the
module reference annotation, the environment distinguishes between
self references and external references.  When an external module is
referenced ($\mmodvar\neq\mmodvaro$), its value is wrapped in a contract check; a self-reference
is resolved to its (unchecked) value.  This distinction implements the
notion of ``contracts as
boundaries''~\cite{dvanhorn:Findler2002Contracts}, in other words,
contracts are an agreement between the module and its context, and the
module can behave internally as it likes.

\begin{display}[$\mvmod \vdash \mmodvar^\mmodvaro \stdstep
  \mexp$]{Module references}
\[
\begin{array}{l@{\ \vdash\ }l@{\ \stdstep\ }l@{\mbox{ if }}l@{\,}l}
\mvmod & \mmodvar^\mmodvar & \mval & \smod\mmodvar\mcon\mval &\in \mvmod \\
\mvmod & \mmodvar^\mmodvaro & \chk\mcon\mmodvar\mmodvaro\mmodvar\mval & \smod\mmodvar\mcon\mval &\in \mvmod \\ 
\mvmod & \mmodvar^\mmodvaro & \chk\mcon\mmodvar\mmodvaro\mmodvar{\opaque \cdot \mcon} & \smod\mmodvar\mcon\opaque &\in \mvmod
\end{array}
\]
\end{display}

\subsubsection{Contract checking}

With the basic rules handled, we now turn to the heart of the system,
contract checking.  As in section~\ref{sec:scpcf}, as computation is
carried out, we can discover properties of values that may be useful
in subsequently avoiding spurious contract errors.  Our primary
mechanism for remembering such discoveries is to add properties,
encoded as contracts, to values as soon as the computational process
proves them.  If a value passes a flat contract check, we add the
checked contract to the value's remembered set.  Subsequent checks of
the same contract are thus avoided.  We divide contract checking
reductions into two categories, those for flat contracts and those for
higher-order contracts, and consider each in turn.

\paragraph{Flat contracts:}  First,
checking flat contracts is handled by three rules, presented in
figure~\ref{fig:flatcon}, depending on whether the value has already
passed the relevant contract.

\begin{figure}
\begin{displayfig}[$\achksimple\mcon\mval\;\stdstep\;\mexpo$]{Flat contract reduction}
\[
\begin{array}{r@{\ }c@{\ }lr}
\chk\mcon\mmodvar\mmodvaro\mmodvaroo\mval
&\stdstep&
\mval\cdot\mcon
&{\mbox{if }\mcon\mbox{ is flat}
\mbox{ and }\proves\mval\mcon}
\\[1mm]
\chk\mcon\mmodvar\mmodvaro\mmodvaroo\mval
&\stdstep&
{\simpleblm\mmodvar\mmodvaroo}
&{\mbox{ if }\mcon\mbox{ is flat}
\mbox{ and }\refutes\mval\mcon}
\\[1mm]
\chk\mcon\mmodvar\mmodvaro\mmodvaroo\mval
&\stdstep&
\multicolumn{2}{l}{\sif{(\sapp{\fc(\mcon)}\mval)}{(\mval\cdot\mcon)}{\simpleblm\mmodvar\mmodvaroo}}
\\
\multicolumn{4}{r}{
\mbox{if }\mcon\mbox{ is flat and }\ambig\mval\mcon}
\\
\end{array}
\]
\end{displayfig}
\begin{displayfig}[$\fc(\mcon) = \mexp$]{Flat contract checking}
\[
\begin{array}{@{}r@{\ }c@{\ }l@{}}
\fc(\srecc\mvar\mcon) &=& \srec\mvar{{\fc(\mcon)}}
\\
\fc(\mvar) &=& \mvar\\
\fc(\spred{\mexp}) &=& \mexp
\\
\fc(\sandc{\mcon_1}{\mcon_2}) &=&
\slamnp\mvaro{\sand{(\sapp{\fc(\mcon_1)}\mvaro)}{(\sapp{\fc(\mcon_2)}\mvaro})}
\\
\fc(\sorc{\mcon_1}{\mcon_2}) &=&
\slamnp\mvaro{\sor{(\sapp{\fc(\mcon_1)}\mvaro)}{(\sapp{\fc(\mcon_2)}\mvaro})}
\\
\fc(\sconsc{\mcon_1}{\mcon_2}) &=&\\
\multicolumn{3}{r}{
\slamnp\mvaro{(\texttt{and}\:{(\sapp\sconsp\mvaro)}{\:}{(\sapp{\fc(\mcon_1)}{(\sapp\scar\mvaro)})}\:{(\sapp{\fc(\mcon_2)}{(\sapp\scdr\mvaro)})})}}
\end{array}
\]
\end{displayfig}
\caption{Flat contracts}
\label{fig:flatcon}
\end{figure}

The first two rules consider the case where the value definitely does
pass the contract, written $\proves\mval\mcon$ (``$\mval$ proves
$\mcon$''), or does not pass, written $\refutes\mval\mcon$ (``$\mval$
refutes $\mcon$'').
If neither of these is the case, written $\ambig\mval\mcon$,
the third rule implements a contract check by compiling it to an
\syntax{if}-expression.  The test is an application of the function
generated by $\fc(\mcon)$ to $\mval$.  If the test succeeds,
$\mval\cdot\mcon$ is produced. Otherwise, the positive party,
here $\mmodvar$, is blamed for breaking the contract on
$\mmodvaroo$.

The three judgments checking the relation between values and contracts
are a simple proof system; by parameterizing over these relations, we
enable our system to make use of sophisticated existing decision
procedures.  For the moment, the key property is that
$\proves{\mval\cdot\mcon}\mcon$ holds, just as in section~\ref{sec:checking-scpcf}, and further details are discussed in section~\ref{sec:proof-system}.

\paragraph{Compiling flat checks to predicates:}

The $\fc$ metafunction, also in figure~\ref{fig:flatcon}, takes a flat contract and produces
the source code of a function which when applied to a value produces
true or false indicating whether the value passes the contract.  The
additional complexity over the similar rules of
sections~\ref{sec:cpcf} and~\ref{sec:scpcf}
handles the addition of flat contracts containing recursive contracts,
disjunctive and conjunctive contracts, and pair contracts.  In particular, to check disjunctive contracts,
we must \emph{test} if the left disjunct passes the contract, and conditionalize
on the result, whereas our earlier reduction rules for flat contracts simply
 \emph{fail} for contracts that don't pass.
As an example, the  check expression
$\chk{\spred\snatp}\mmodvar\mmodvaro\mmodvaroo\mval$ reduces to
$\sif{(\sapp{\snatp}\mval)}\mval{\simpleblm\mmodvar\mmodvaroo}$,
but using this reduction to check $\achksimple{\sorc{\spred\snatp}{\spred\sboolp}}{\strue}$
would cause a blame error when checking the left disjunct, which is
obviously not the intended result.  Instead, the rules for $\fc$
generate the check
$\sif{(\sapp\snatp\strue)}{\strue}{(\achksimple{\spred\sboolp}\strue)}$,
which succeeds as desired.

\paragraph{Higher-order contracts:}
The next set of reduction rules, presented in figure~\ref{fig:hocon},
defines the behavior of higher-order contract checks; we assume for
these rules that
the checked contract is not flat.

\begin{figure}
\begin{displayfig}[$\achksimple{\mcon}\mval\,\stdstep\,\mexpo$]{Function
    contract reduction}
\[
\begin{array}{@{\,}l@{\;}c@{\;}l@{\;\;}l}
\chk{\sdep{\mcon}\mvar{\mcono}}\mmodvar\mmodvaro\mmodvaroo\mval
& \stdstep &
\\
\multicolumn{4}{r}{
\slam\mvar{\chk{\mcono}\mmodvar\mmodvaro\mmodvaroo{(\sapp\mval{\chk{\mcon}\mmodvaro\mmodvar\mmodvaroo\mvar})}}}
\\
&&&{\mbox{ if }\deltamap\sprocp\mval\strue}
\\[1mm]
\chk{\sdep{\mcon}\mvar{\mcono}}\mmodvar\mmodvaro{\mmodvaroo}\mval
& \stdstep &
{\ablm\mmodvar{\mmodvaroo}{\relax}{\sarr{\mcon}{\mcono}}\mval}
& {\mbox{ if }\deltamap\sprocp\mval\sfalse}
\\
\end{array}
\]
\end{displayfig}
\begin{displayfig}[]{Other higher-order contract reductions}
\[
\begin{array}{@{}l@{\;}c@{\;}l@{}r}
\achksimple{\sconsc{\mcon}{\mcono}}{\mval}
&\stdstep&\\
\multicolumn{4}{r}{
\scons{\achksimple{\mcon}{\sapp\scar{\mval'}}}
         {\achksimple{\mcono}{\sapp\scdr{\mval'}}}}
\\
\multicolumn{4}{r@{}}{
\mbox{if }\deltamap\sconsp\mval\strue\mbox{ and }\mval' = \mval\cdot{\liftpred\sconsp}}
%
\\[1mm]
\chk{\sconsc{\mcon}{\mcono}}\mmodvar\mmodvaro{\mmodvaroo}\mval
& \stdstep
& \ablm\mmodvar{\mmodvaroo}{\relax}{\sarr{\mcon}{\mcono}}\mval
&\mbox{if }\deltamap\sconsp\mval\sfalse
\\[1mm]
\achksimple{\srecc{\mvar}{\mcon}}\mval
&\stdstep&
\multicolumn{2}{@{}l}{\achksimple{\subst{\srecc{\mvar}{\mcon}}\mvar\mcon}\mval}
\\[1mm]
\achksimple{\sandc\mcon\mcono}{\mval}
&\stdstep&
\multicolumn{2}{@{}l}{\achksimple\mcono{\achksimple\mcon\mval}}
\\[1mm]
\achksimple{\sorc\mcon\mcono}{\mval}
&\stdstep&
\multicolumn{2}{@{}l}{\sif{(\sapp{\fc(\mcon)}\mval)}{(\mval\cdot\mcon)}{\achksimple\mcono\mval}}
\\
\multicolumn{4}{r}{
\mbox{if }\ambig\mval\mcon}
\\[1mm]
\achksimple{\sorc\mcon\mcono}\mval
&\stdstep&
\mval
&\mbox{ if }\proves\mval\mcon
\\[1mm]
\achksimple{\sorc\mcon\mcono}\mval
&\stdstep&
{\achksimple\mcono\mval}&
{\mbox{ if }\refutes\mval\mcon}
\\[1mm]
%
\end{array}
\]
\end{displayfig}
\caption{Higher-order contract reduction}
\label{fig:hocon}
\end{figure}

In the first rule, we again use the $\eta$-expansion technique
pioneered by \citet{dvanhorn:Findler2002Contracts} to decompose a
higher-order contract into subcomponents.  This rule only applies if
the contracted value $\mval$ is indeed a function, as indicated by
$\syntax{proc?}$ (In SCPCF, this side-condition is avoided thanks to
the type system).  Otherwise,
the second rule blames the positive party of a function contract when
the supplied value is not a function.

The remaining rules handle higher-order contracts that are not
immediately function contracts, such as pairs of function contracts.
The first two are for pair contracts.  If the value is determined to
be a pair by $\sconsp$, then the components are extracted using
$\scar$ and $\scdr$ and checked against the relevant portions of the
contract.  If the value is not a pair, then the program reduces to
blame, analogous to the case for function contracts.

The last set of rules decompose combinations of higher-order
contracts.  Recursive contracts are unrolled.  (Productivity ensures
that contracts do not unroll forever.) Conjunctions are split into
their components, with the left checked before the right.  For
higher-order disjunctions, we rely on the invariant that only the
right disjunct is higher-order and use $\fc$ for the check of the
left.  When possible, we omit the generation of this check by using
the proof system as described above (in the final two rules).

\subsubsection{Applying abstract values}

Again, application of abstract values poses a challenge, just as it in
in section~\ref{sec:scpcf}.   We now must explore more possible
behaviors of abstract operators, and we no longer have types to guide
us.  Fortunately, abstract values also give us the tools to express
the needed computation.

\begin{display}[$\mexp  \ \stdstep \ \mexpo \mbox{ {\rm where} }  \ab\mval\,=\,\with\opaque\mconset$]{Applying abstract values}
\[
\begin{array}{rclr}
\sapp{\ab\mval}\mval'
&\stdstep&
\multicolumn{2}{l}{\with\opaque{\{
\subst{\ab\mval}\mvar\mcono\ |\ (\sdep\mcon\mvar\mcono) \in \mconset \}}}
\\
&&&\mbox{if }\deltamap\sprocp{\ab\mval}\strue
\\
\sapp{\ab\mval}\mval'
&\stdstep&
\sapp{\syntax{havoc}}\mval' 
&\mbox{if }\deltamap\sprocp{\ab\mval}\strue
\end{array}
\]
\[
\begin{array}{c}
\syntax{havoc}=
\srec{\syntax{y}}{\slam\varx{\textsc{amb}(\{
\sapp{\syntax{y}}{(\sapp\varx\opaque)},
\sapp{\syntax{y}}{(\sapp\scar\varx)},
\sapp{\syntax{y}}{(\sapp\scdr\varx)}\})}}
\end{array}
\]
\[
\begin{array}{lcl}
\textsc{amb}(\{\mexp\}) &=& \mexp\\
\textsc{amb}(\{\mexp,\mexp_1,\dots\}) &=& \sif\opaque\mexp{\textsc{amb}(\{\mexp_1,\dots\})}
\end{array}
\]
\end{display}

The behavior of abstract values, which are created by references
opaque modules, is handled in much the same way as in SCPCF.  When an
abstract function is applied, there are again two possible scenarios:
(1) the abstract function returns an abstract value or (2) the
argument escapes into an unknown context that causes the value to be
blamed.  We again make use of a $\syntax{havoc}$ function for
discovering if the possibility of blame exists. In contrast to the
typed setting of SCPCF, we need only one such value.  The demonic
context is a universal context that will produce blame if it there
exists a context that produces blame originating from the value.  If
the universal demonic context cannot produce blame, only the range
value is produced.

The $\syntax{havoc}$ function is implemented as a recursive function
that makes a non-deterministic choice as to how to treat its
argument---it either applies the argument to the least-specific value,
$\opaque$, or selects one component of it, and then recurs on the
result of its choice.  This subjects the input value to all possible
behavior that a context might have.  Note that the demonic context
might itself be blamed; we implicitly label the expressions in the
demonic context with a distinguished label and disregard these
spurious errors in the proof of soundness.
We use the \textsc{amb} metafunction to implement the non-determinism
of \syntax{havoc}; \textsc{amb} uses an \syntax{if} test of an opaque
value, which reduces to both branches.

\subsection{Proof system}
\label{sec:proof-system}

Compared to the very simple proof system of
section~\ref{sec:checking-scpcf}, the system for proving or refuting
whether a given value satisfies a contract in Core Racket is 
more sophisticated, although the general principles remain the same.

In particular, we rely on three different kinds of judgments that
relate values and contracts: proves, refutes, and neither.  The first,
$\proves\mval\mcon$ includes the original judgment that a value proves
a contract if it remembers that contract.  Additionally, we add
judgements for reasoning about type predicates in the language.  For
example if a value is known to satisfy a particular base predicate,
written $\proves\mval\moppred$, then the value satisfies the contract
$\liftpred\moppred$.  This relies on the relation between values and
predicates used above in the definition of $\delta$, which is defined
in a straightforward way.

The refutes relation is more interesting and relies on additional
semantic knowledge, such as the disjointness of data types.  For
instance, a value that remembers it is a procedure, refutes all pair
contracts and the \syntax{pair?} predicate contract.  Other refutes
judgments are straightforward based on structural decomposition of
contracts and values.

The complete definition of these relations is given in
appendix~\ref{sec:absdelta}.  Our implementation, described in
section~\ref{sec:implementation}, incorporates a richer set of rules
for improved reasoning.  The implementation is naive but effective for
basic semantic reasoning, however it essentially does no sophisticated
reasoning about base type domains such as numbers, strings, or lists.
The tool could immediately benefit from leveraging an external solver
to decide properties of concrete values.

\subsection{Improving precision via non-determinism}

Since our reduction rules, and in particular the $\delta$ relation,
make use of the remembered contracts on values, making these contracts
as specific as possible improves precision of the results.

\begin{display}[$\ab\mval \ \stdstep \ \ab\mval'$]{Improving precision via non-determinism}
\[
\begin{array}{lcl}
\with\opaque{\mconset\dotcup\{ \sorc{\mcon_1}{\mcon_2} \}}
&\stdstep&
\remcon\opaque{\mconset\cup{\{ \mcon_i \}}}\quad i \in \{1,2\}
\\[1mm]
\with\opaque{\mconset\dotcup\{ \srecc\mvar\mcon \}}
&\stdstep&
\remcon\opaque{\mconset\cup{\{ \subst{\srecc\mvar\mcon}\mvar\mcon \}}}
\end{array}
\]
\end{display}
The two rules above increase the specificity of abstract values.  The
first splits abstract values known to satisfy a disjunctive contract.
For example,
$\with\opaque\{\sorc{\liftpred\snatp}{\liftpred\sboolp}\}\: \stdstep\:
\with\opaque{\liftpred\snatp}$ and $\with\opaque{\liftpred\sboolp}$.
This converts the imprecision of the value into non-determinism in
the reduction relation, and makes subsequent uses of $\delta$ more
precise on the two resulting values.  Similarly, we unfold recursive
contracts in abstract values; this exposes further disjunctions to
split, as with a contract for lists.

As an example of the effectiveness of this simple approach, consider
the list length function:
\begin{alltt}
(module length
  (provide [len (list/c -> nat?)])
  (define len
    (\(\uplambda\) (l)
      (if (empty? l) 0 (+ 1 (len (cdr l)))))))
\end{alltt}
When applied to the symbolic
value $$\opaque\cdot\srecc\varx{(\sorc{\liftpred{\syntax{empty?}}}{\sconsc\snatp{\varx}})}$$
which is the definition of \syntax{list/c},
we immediately unroll and split the abstract value, meaning that we
evaluate the body of $\syntax{len}$ in exactly the two cases it is
designed to handle, with a precise result for each.  Without this
splitting, the test would return both $\strue$ and $\sfalse$, and the
semantics would attempt to take the $\syntax{cdr}$ of the empty list,
even though the function will never fail on concrete inputs.  This
provides some of the benefits of occurrence
typing~\cite{dvanhorn:TobinHochstadt2010Logical} simply by exploiting
the non-determinism of the reduction semantics.


\subsection{Evaluation and Soundness}
\label{sec:racket-soundness}

We now define evaluation of entire modular programs, and prove soundness for
our abstract reduction semantics.  One complication remains.  In any program with
opaque modules, any module might be referenced, and then treated
arbitrarily, by one of the opaque modules.  While this does not
affect the value that the main expression might reduce to, it does
create the possibility of blame that has not been previously
predicted.  We therefore place each concrete module into the
previously-defined demonic context and non-deterministically choose
one of these expressions to run \emph{prior} to running the main
module of the program.

The evaluation function is defined as:
\[
\mathit{eval}(\mvmod\mexp) =
  \{ \mexp'\ |\ \mvmod \,\vdash\, \sbegin\mexpo\mexp \multistdstep \mexp'\}\text,
\]
where $\mexpo =
\textsc{amb}(\{\strue,\overline{\sapp{\syntax{havoc}}{\mmodvar}}\})$,
$\smod\mmodvar\mcon\mval \in \mvmod$.

Soundness, as in section~\ref{sec:soundness}, relies on the definition
of approximation between terms, and its straightforward extension to
modules and programs.  (The details of the approximation relation for
Symbolic Core Racket are given in appendix \ref{sec:proofs}.)

\begin{theorem}[Soundness of Symbolic Core Racket]
\label{thm:soundness}
\ \\
If $\mprg \sqsubseteq \mprgo$ where $\mprgo = \mvmod
\mexp$ and $\mans \in \mathit{eval}(\mprg)$, then there exists
 some $\manso \in \mathit{eval}(\mprgo)$ where $\mans \sqsubseteq_\mvmod \manso$.
\end{theorem}

\noindent
This soundness result, proved in appendix~\ref{sec:proofs}, implies
that if a program with opaque modules does not produce blame, then the
known modules cannot be blamed, \emph{regardless} of the choice of
implementation for the opaque modules.

\begin{corollary}
If  $\mmodvar$ is the name of a concrete module in $\mprg$, and $\mprg
\not\multistdstep \mbox{\emph{$\mctx[\simpleblm\mmodvar\mmodvaro]$}}$,
 then \emph{no} instantiation of the opaque modules in
$\mprg$ can cause $\mmodvar$ to be blamed.
\end{corollary}

\section{Convergence and decidability}
\label{sec:machine}

At this point, we have constructed an abstract reduction semantics
that gives meaning to programs with opaque components.  The semantics
is a sound abstraction of all possible instantiations of the omitted
components, thus it can be used to verify modular programs satisfy
their specifications.  However, in order to automatically verify
programs, the semantics must converge for the particular program being
analyzed.

We now describe how to refactor the semantics in such a way that we
can accelerate and---if desired---\emph{guarantee} convergence by
introducing further orthogonal approximation into the semantics.  This is
accomplished in a number of ways:
\squishlist
\item basic operations may widen when applied to concrete values,

\item environment structure may be bounded, and

\item control structure may be bounded.
\squishend

In order to guarantee convergence for all possible programs, all
three of these forms of approximation must be employed and are sufficient to
guarantee decidability of the semantics.
However, our experience suggests that such strong convergence
guarantees may be unnecessary in practice.  For example, we have found
that a simple widening
of  concrete recursive function applications to
  their contract when applied to abstract values works well for ensuring
convergence of tail-recursive programs broken into small
modules.
By adding a limited form of control structure approximation, we are
able to automatically verify non-tail-recursive functions.
Taken together, these forms of approximation do not guarantee
convergence in general, yet they do induce convergence for all of the
examples we have considered (see \S\ref{sec:examples}) and with fewer
spurious results compared to more traditional forms of abstraction
such as 0CFA and pushdown flow analysis.
But rather than advocate a particular approximation strategy, we now
describe how to refactor the semantics so that all these choices
may be expressed. Our implementation (\S\ref{sec:implementation})
then makes it easy to explore any of them.

\subsection{Widening values}

 Widening the results of basic
operations, i.e., those interpreted by $\delta$, can cut down
the set of base values, and if necessary can ensure finiteness of base
values.  Thus, we replace $\delta$ with $\delta'$:
\[
\delta'(\mop,\vec\mval) \ni \widen(\mval) \iff
\delta(\mop,\vec\mval) \ni \mval
\]
where $\widen$ represents an arbitrary choice of a metafunction for
mapping a value to its approximation.  To avoid approximation, it is
interpreted as the identity function.  To ensure finite base values,
it must map to a finite range; a simple example is $\widen(\mval) =
\bullet$ for all $\mval$.  An example of a more refined interpretation
is $\widen(n) = \spred\snatp$, $\widen\scons\mval\mvalo =
\spred\sconsp$, etc.  For soundness, we require that $\widen(\mval) =
\mvalo$ implies $\mval \sqsubseteq \mvalo$.

\subsection{Bounding environment structure}

 The lexical environment of a
program represents a source of unbounded structure.  To enable
approximation of the lexical environment, we first refactor the
semantics as calculus of explicit
substitutions~\cite{dvanhorn:Curien1991Abstract} with a global store.
Substitutions are modeled by finite maps from variables to addresses
and the store maps addresses to \emph{sets of} values, which are now
represented as closures:
\[
\begin{array}{rcl}
\menv,\menvo &::=& \emptyset\ |\ \menv[\mvar\mapsto\maddr]\\
\msto,\mstoo &::=& \emptyset\ |\ \msto[\maddr\mapsto \{\mval,\dots\}]
\end{array}
\]
Reductions that bind variables, such as function application, must allocate
and extend the environment.  For example,
\[
\begin{array}{r@{\quad}c@{\quad}l@{}l}
\sapp{\slam\mvar\mexp}\mval^\mlab
&\stdstep&
\subst\mval\mvar\mexp
\end{array}
\]
becomes an analogous reduction relation on closures and stores:
\[
\begin{array}{r@{\quad}c@{\quad}l@{}l}
\sapp{(\slam\mvar\mexp,\menv)}\mval^\mlab,\msto
&\stdstep&
(\mexp,\menv[\mvar\mapsto \maddr]),\msto\sqcup[\maddr\mapsto\mval]
\\
&&\mbox{where }\maddr = \alloc{\msto,\mvar}
\end{array}
\]
and the interpretation of $\msto\sqcup[\maddr\mapsto\mval]$ is
$\msto'$ s.t. $\msto'(\maddro) = \msto(\maddro)$ if $\maddr \neq \maddro$
and $\msto'(\maddr) = \msto(\maddr) \cup \{\mval\}$.

Since reduction operates over closures, there is an additional case
needed to handle variable references:
\[
\begin{array}{r@{\quad}c@{\quad}l@{}l}
(\mvar,\menv), \msto &\stdstep&
\mval, \msto \mbox{ if } \mval \in \msto(\menv(\mvar))
\end{array}
\]

The $\allocname$ metafunction provides a point of control which
regulates the approximation of environment structure.  To ensure
finite environment approximation, the metafunction must map to a
finite set of addresses (for a fixed program).  A simple finite
abstraction is $\alloc{\msto,\mvar} = 0$ for all $\msto$ and $\mvar$.
This abstraction maps all bindings to a single location, thus
conflating all bindings in a program.  Although highly imprecise, this
is a sound approximation, and in fact \emph{any} instantiation of
$\allocname$ is sound~\cite{dvanhorn:VanHorn2010Abstracting}.  A more refined abstraction is
$\alloc{\msto,\mvar} = \mvar$, which provides a finite abstraction of
environment structure similar to 0CFA in which multiple bindings of
the same variable are conflated.  To avoid approximation,
$\alloc{\msto,\mvar}$ should chose a fresh address not in $\msto$.
Consequently, the store maps all addresses to singleton sets of values
and the environment-based reduction semantics corresponds precisely with
the original.

\subsection{Bounding control structure}
 The remaining source of unbounded
structure stems from the control component of a program.  Bounding the
environment structure was achieved by (1) making substitutions
explicit as environments and (2) threading environments through a
store which could be bounded.  An analogous approach is taken for
control: (1) evaluation contexts are explicated as continuations and
(2) continuations are threaded through the store.

The resulting semantics is an abstract machine that
operates over a triple comprised of a closure, a store, and a
continuation.  Transitions take three forms: decomposition steps,
which search for the next redex and push continuations if needed; plug
steps which return a value to a context and pop continuations if
needed; and contraction steps, which implement the $\sstep$ notion of
reduction.

We write continuations $\mcont$ as single evaluation context frames with
embedded addresses representing (a pointer to) the surrounding context.
So for example $\sapp\mexp\maddr$ represents $\sapp\mexp\mctx$ where
$\maddr$ points to a continuation representing $\mctx$.
A simple decompose case is:
\[
\begin{array}{r@{\quad}c@{\quad}l@{}l}
\langle (\sapp\mexp\mexpo,\menv),\msto,\mcont\rangle
&\longmapsto &
\langle (\mexp,\menv),\msto \sqcup [\maddr\mapsto\mcont],\sapp\maddr{(\mexpo,\menv)}\rangle
\\
&&\mbox{where }\maddr = \alloc{\msto,\mcont}
\end{array}
\]
Notice that this transition searches for the next redex in the left
side of an application and allocates a pointer to the given context in
order to push on the argument to be evaluated later.  Similar to
variable binding, continuations are allocated using $\allocname$ and
joined in the store, allowing for multiple continuations to reside in
a single location.  The corresponding plug rule pops the current
continuation frame and non-deterministically chooses a continuation:
\[
\begin{array}{r@{\quad}c@{\quad}l@{}l}
\langle \mval,\msto,\sapp\maddr{(\mexpo,\menv)}\rangle
&\longmapsto &
\langle \sapp\mval{(\mexpo,\menv)},\msto,\mcont\rangle
\mbox{ where }\mcont \ni \msto(\maddr)
\end{array}
\]
The contraction rule simply applies the reduction relation on explicit
substitutions:
\[
\begin{array}{r@{\quad}c@{\quad}l@{}l}
\langle (\mexp,\menv),\msto,\mcont \rangle
&\longmapsto &
\langle (\mexpo,\menvo),\mstoo,\mcont \rangle\\
\multicolumn{3}{r@{}}{\mbox{ if } ((\mexp,\menv),\msto)\,\stdstep\,((\mexpo,\menvo),\mstoo)}
\end{array}
\]

To ensure a finite approximation of control, $\allocname$ must map to
a finite set of addresses.  A simple finite abstraction is
$\alloc{\mcont,\msto} = 0$.  A more refined finite abstraction of
control is to use a frame abstraction: $\alloc{\sapp\mexp\maddr,\msto}
= \sapp\mexp{[\;]}$ and likewise for other continuation forms.  To
avoid approximation, $\alloc{\mcont,\msto}$ should produce an address
not in $\msto$.

We have now restructured our semantics as a machine model with three
distinct points of control over approximation: basic operations,
environments, and control; the full definition of the machine is given
in appendix~\ref{full-machine}. We now establish the correspondence
between the previous reduction semantics and the machine model when no
approximation occurs.  Let $\stdstep_{\mathit{CESK}}$ denote the
machine transition relation under the exact interpretations of
$\widen$ and $\allocname$.  Let $\unload$ be the straightforward
recursive ``unload'' function that maps a closure and store to the
closed term it represents.
\begin{lemma}[Correspondence]
If $\mprg \multistdstep \mprgo$, then there exists $\mstate$ such that
$\langle\mprg,\varnothing,\varnothing\rangle
\multistdstep_{\mathit{CESK}} \mstate$ and $\unload(\mstate) =
\mprgo$.
\end{lemma}

We now relate any approximating variant of the machine to its
exact counterpart.  Let $\stdstep_{\widehat{\mathit{CESK}}}$ denote
the machine transition under any sound interpretation of $\widen$.  We
define an \emph{abstraction map} as a structural abstraction of the
state-space of the exact machine to its approximate counterpart.  The
key case is on stores:
\begin{align*}
  \alpha(\msto) &= \lambda \hat{\maddr} . \!\!\!\! \bigsqcup_{\alpha(\maddr) = \hat\maddr} \!\!\!\! \{\alpha(\msto(\maddr))\}
\end{align*}
The $\sqsubseteq$ relation is lifted to machine states as the natural
point-wise, element-wise, component-wise and member-wise lifting.

\begin{theorem}[Soundness]
\label{thm:soundness-machine}
If $\mstate \stdstep_{\mathit{CESK}} \mstate'$ and $\alpha(\mstate)
\sqsubseteq \hat\mstate'$, then there exists $\hat\mstate'$ such that
$\hat\mstate \stdstep_{\widehat{\mathit{CESK}}} \hat\mstate'$ and
$\alpha(\mstate') \sqsubseteq \hat\mstate'$.
\end{theorem}
We have now established any instantiation of the machine is a sound
approximation to the exact machine, which in turn correspond to the
original reduction semantics.  Furthermore, we can prove decidability
of the semantics for finite instantiations of $\widen$ and $\allocname$:
\begin{theorem}[Decidability]
\label{thm:decidability}
If $\widen$ and $\allocname$ have finite range for a
program $\mprg$, then $\langle\mprg,\varnothing,\varnothing\rangle
\multistdstep_{\widehat{\mathit{CESK}}} \mstate$ is decidable for any $\mstate$.
\end{theorem}

The proofs of these theorems closely follow those given by
\citet{dvanhorn:VanHorn2010Abstracting} and are deferred to appendix~\ref{sec:proofs}.

\section{Implementation}
\label{sec:implementation}

\begin{figure}
\hspace{-4mm}
\includegraphics[width=1.1\columnwidth]{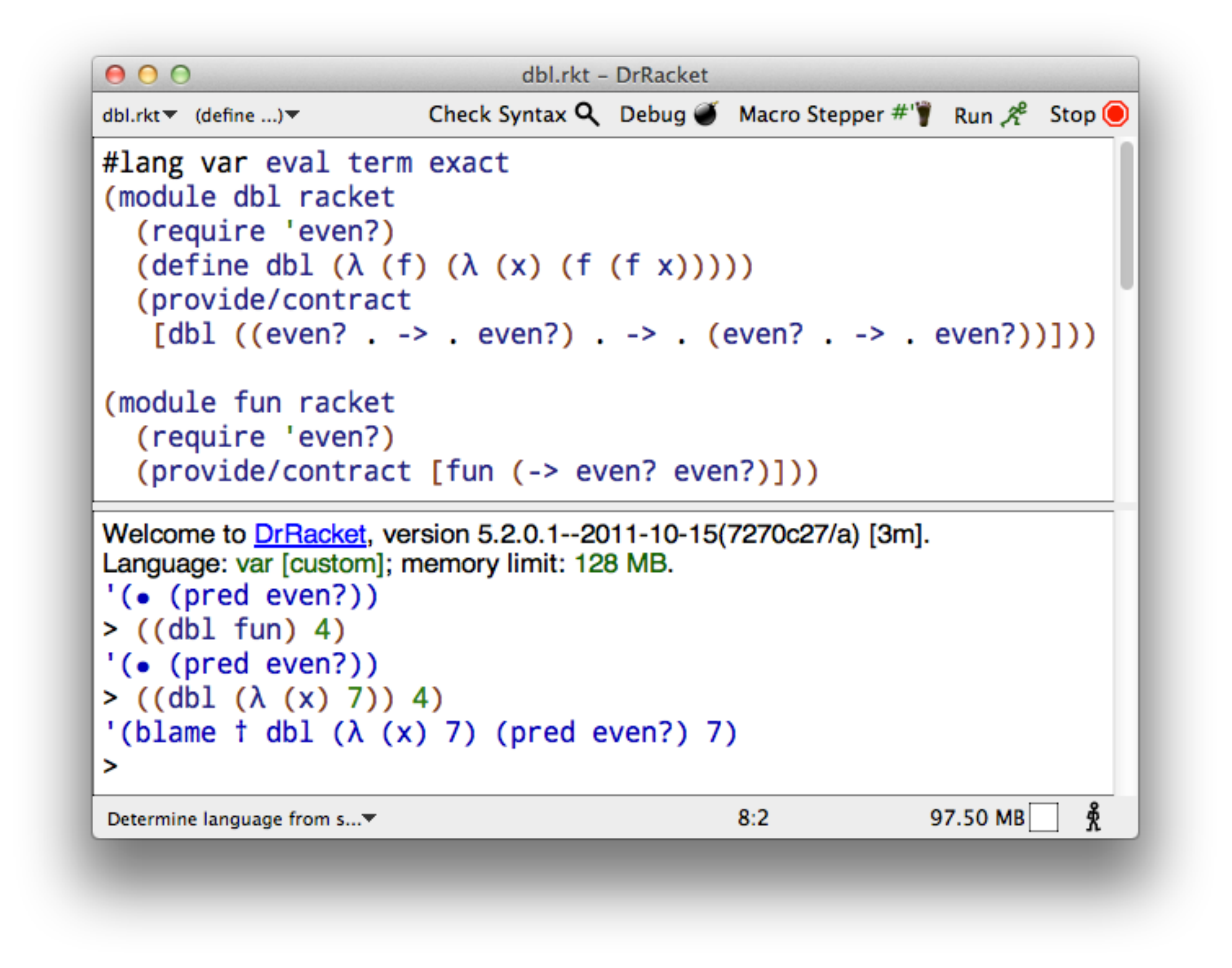}
\vspace{-6mm}
\caption{Interactive program verification environment}
\label{fig:screenshot}
\end{figure}

To validate our approach, we have implemented a prototype interactive
program verification environment, as seen in
figure~\ref{fig:screenshot}.
We can take the example from section~\ref{sec:tools}, define
the relevant modules, and explore the behavior of different choices
for the main expression.

Programs are written with the {\tt \#lang var \emph{<options>}}
header, where {\tt \emph{<options>}} range over a visualization mode:
{\tt trace}, {\tt step}, or {\tt eval}; a model mode: {\tt term} or
{\tt machine}; and an approximation mode: {\tt approx}, {\tt exact},
or {\tt user}.  Following the header, programs are written in a subset
of Racket, consisting of a series of module definitions and a
top-level expression.

The visualization mode controls how the state space is explored. The
choices are simply running the program to completion with a
read-eval-print loop, visualizing a directed graph of the state space
labeled by transitions, or with an interactive step-by-step
exploration.  The model mode selects whether to use the {\tt term} or
the {\tt machine} as the underlying model of computation.

Finally, the approximation mode selects what, if any, approximation
should be used.  The {\tt exact} mode uses no approximation;
allocation always returns fresh addresses and no widening is used for
base values.  The {\tt approx} mode uses a default mode of
approximation that has proved useful in verifying programs (discussed
below).  The {\tt user} mode allows the user to provide their own
custom metafunctions for address allocation and widening.

\subsection{Implementation extensions}

Our prototype includes significant extensions to the system as
described above.

First, we make numerous extensions in order to
verify existing Racket programs.  For example, modules are extended to
include multiple definitions and functions may accept zero or more
arguments.  The latter complicates the reduction relation as new
possibilities arise for errors due to arity mismatches.
Second, we add additional base values and operations to the
  model to support more realistic programs.
  Third, we make the implementation of contract checking and
  reduction more sophisticated, to reduce complexity in simple cases,
  improving running time and simplifying visualization.
Fourth, we implement several techniques to reduce the size of the
state space explored in practice, including abstract garbage
  collection~\cite{dvanhorn:Might:2006:GammaCFA}.  Abstract GC
  enables naive allocation strategies to perform with high precision.
  Additionally, we widen contracted recursive functions to their
  contracts on recursive calls; this implements a form of induction
  that is highly effective at increasing convergence.
Fifth, we add simpler rules to model non-recursive functions and
  non-dependent contracts.  This brings the model closer to
  programmers expectation of the semantics of the language, and
  simplifies visualizations.
Sixth, we include several more contract combinators, such as
  {\tt and/c} for contract conjunction; {\tt atom/c} for expressing
  equality with atomic values; {\tt one-of/c} for finite enumerations;
  {\tt struct/c} for structures; and {\tt listof} and {\tt
    non-empty-listof} for lists of values..

Finally, we provide richer blame information when programs go
  wrong, as can be seen in the screen shot of
  figure~\ref{fig:screenshot}.  Our system reports the full complement of
  information available in Racket's production contract library, which
  reports the failure of \texttt{dbl} as:
\begin{alltt}
> ((dbl (\(\uplambda\!\) (x) 7)) 4)
\emph{top-level broke (even? \(\rightarrow\) even?) \(\rightarrow\)}
\emph{\qquad\qquad\qquad\qquad(even? \(\rightarrow\) even?)
on dbl; expected <even?>, given: 7}
\end{alltt}

\noindent
Our prototype is available  at\;
\href{https://github.com/samth/var/}{\texttt{github.com/samth/var}}\;.

\subsection{Verified examples}
\label{sec:examples}

We have verified a number of example programs, which fall into three
categories:
\begin{itemize}
\item small programs with rich contracts such as the example of {\tt
    insertion-sort}  from the introduction, where verifying contract
  correctness is close to full functional verification;
\item tricky-to-verify programs with simple contracts such
  as the tautology checker described by Wright and Cartwright, where
  our tool proves not only that the function satisfies its
  \texttt{boolean? -> boolean?} contract, but also that it has no
  internal run-time type errors; and
\item several larger graphical, interactive video games developed
  according to the program-by-design~\cite{dvanhorn:Felleisen2001How} approach that stresses
  data- and contract-driven program design.
\end{itemize}

In the third category, we were able to
automatically verify contract correctness for non-trivial existing
programs, including two, Snake and Tetris, developed
for our undergraduate programming course.  The programs make use of
higher-order functions such as folds and maps and construct anonymous
(upward) functions.  Their contracts include finite enumerations,
structures, and recursive (ad-hoc) unions.  Here is the key
contract definition for the Snake game:
\begin{verbatim}
(define-contract snake/c
  (struct/c snake
    (one-of/c 'up 'down 'left 'right)
    (non-empty-listof
      (struct/c posn nat? nat?))))
\end{verbatim}
We are able to automatically verify these stated invariants, such
as that snakes have at least one segment and that it moves in
one of four possible directions.

For these examples, we found that simply widening the results of recursive
calls with abstract arguments is sufficient to ensure convergence in
the semantics given an appropriate fine-grained module decomposition.
All of our verified examples are available with our implementation.

\section{Related work}
\label{sec:related-work}

The analysis and verification of programs and specifications has been
a research topic for half a century; we survey only closely related
work here.

\paragraph{Symbolic execution}
Symbolic execution~\cite{samth:symexec} is the idea of running a
program, but with abstract inputs.  The technique can be used either
for testing, to avoid the need to specify certain test data, or for
verification and analysis.  Over the past 35 years, it has been used
for numerous testing and verification tasks. There has been a
particular upsurge in interest in the last ten
years~\cite{samth:Cadar2006EXE,samth:klee08}, as high-performance SAT and SMT
solvers have made it possible to eliminate infeasible paths by
checking large sets of constraints.

Most approaches to symbolic execution focus on abstracting first order
data such as numbers, typically with constraints such as inequalities
on the values.  In this paper, we present an approach to symbolic
execution based on contracts as symbols, which scales
straightforwardly to higher-order values.  Despite this focus on
higher-order values, the remembered contracts maintained by our system
let us constrain symbolic execution to feasible evaluations; using an
external solver to decide relationships such as $\proves\mval\mcon$ is
an important area of future work.

Recently, under the heading of \emph{concolic
  execution}~\cite{samth:cute-concolic,
samth:dart-concolic}, symbolic
execution has been paired with test generation to analyze software
more effectively.  We believe that we could effectively use our system
as the framework for such a system, by nondeterminstically reducing
abstract values to concrete instances.

The only work on higher-order symbolic execution that we are aware of
is by \citet{samth:thiemann-symexp} on eliminating redundant pattern
matches.  \citeauthor{samth:thiemann-symexp} considers only a very
restricted form of symbols: named functions partially applied to
arguments, constructors, and a top value.  This approximation is only
sound for a purely functional language, and thus while we could
incorporate it into our current symbolic model of Racket, further
extensions to handle mutable state rule out the technique.  It is
unclear whether redundancy elimination would benefit from contracts as
symbol.

\paragraph{Verification of first-order contracts}

Over the past ten years, there has been enormous success verifying
modular first-order programs, as demonstrated by tools such as the
SLAM and Spec\# projects~\cite{samth:slam-cacm,dvanhorn:BarnettEA10,
  dvanhorn:Fahndrich2011Static}.  These tools typically operate by
abstracting first-order programs in languages such as C to simpler
systems such as automata or boolean programs, then model-checking the
results for violations of specified contracts.

However, these approaches do not attempt to handle
of the higher-order features of languages such as Racket, Python,
Scala, and Haskell.  For instance, the boolean program abstraction
employed by SLAM~\cite{samth:slam-pldi01} is inherently
first-order: variables can take only boolean values.  Our system, in
contrast, scales to higher-order language features.

Despite the fundamental difference, there are important similarities
between this work and ours. The systems all employ nondeterminism
extensively to reason about unknown behavior, and abstract the
environment by allowing it to take arbitrary actions; as in the
\texttt{havoc} statement in Boogie~\cite{dvanhorn:Fahndrich2011Static}
which we generalize to the \textsf{havoc} function for placing
higher-order functions in an arbitrary context.

Additionally, we believe that the techniques used in these existing
first-order tools could improve
precision for first-order predicate checks in our system; exploring
this is an important avenue for future work.

\paragraph{Verification of higher-order contracts}
The most closely related work to ours is the modular set-based
analysis based on contracts of
\citeauthor{dvanhorn:Meunier2006Modular}~\cite{dvanhorn:Meunier2006Modular,dvanhorn:Meunier2006Diss}
and the static contract checking of
\citeauthor{dvanhorn:Xu2009Static}~\cite{dvanhorn:Xu2009Static,samth:Xu2012:Hybrid}.

\textbf{\citeauthor{dvanhorn:Meunier2006Modular}} take a program analysis approach, generating set
constraints describing the flow of values through the program text.
When solved, the analysis maps source labels to sets of abstract
values to which that expression may evaluate.  Meunier's system is
more limited than ours in several significant ways.

First, the set-based analysis is defined as a separate semantics,
which must be manually proved to correspond to the concrete semantics.
This proof requires substantial support from the reduction semantics,
making it significantly and artificially more complex by carrying
additional information  used only in the proof.  Despite this, the
system is unsound, since it lacks an analogue of \syntax{havoc}.
This unsoundness has been verified in Meunier's prototype.

Second, while our semantics allows the programmer to choose how much
to make opaque and how much to make concrete, Meunier's system always
treats the entire rest of the program opaquely from the perspective of
each module.

Third, our language of contracts is much more expressive: we consider
disjunction and conjunction of contracts, dependent function contracts,
and data structure contracts.  Our ability to statically reason about
contract checks is significantly greater---Meunier's system
includes only the simplest of our rules for $\proves\mval\mcon$.

Finally, Meunier approximates conditionals by the union of its
branches' approximation; the test is ignored.
This seemingly minor point becomes significant when considering
predicate contracts.  Since predicate contracts reduce to
conditionals, this effectively  approximates all predicates as both
holding and not holding, and thus \emph{all predicate contracts may
  both fail and succeed}.

\textbf{\citeauthor{dvanhorn:Xu2009Static}} \cite{dvanhorn:Xu2009Static} describe a static
contract verification system for Haskell.  Their approach is to
compile contract checks into the program, using a transformation
modeled on \citet{dvanhorn:Findler2002Contracts}, simplify the program
using the GHC optimizer, and examine the result to see if any contract
checks are left in the residual program.  In subsequent work,
\citet{samth:Xu2012:Hybrid} applies this approach to OCaml, providing
a formal account of the simplifier employed, and extending
simplification by using an SMT solver as an oracle for some
simplification steps.  In both systems, if not all contract checks are
eliminated by simplification, the system reports them as potentially
failing.

Our approach extends that of \citeauthor{dvanhorn:Xu2009Static} in
three crucial ways.  First, our symbolic execution-based approach
allows us to consider full executions of programs, rather than just a
simplification step.  Second, \citeauthor{dvanhorn:Xu2009Static}
considers a significantly restricted contract language, omitting
conjunction, disjunction, and recursive contracts, as well as
contracts that may not terminate, may fail, or are include calls to
unknown functions.  As we saw in section~\ref{sec:core-racket}, these
extensions add significant complexity and expressiveness to the
system.  Third, as with
\citeauthor{dvanhorn:Meunier2006Modular}'s work, the user has no control over what is precisely
analyzed; indeed, \citeauthor{dvanhorn:Xu2009Static} \emph{inline} all non-contracted
functions.

\citet{samth:mcallester-contracts} provide a semantic model of
contracts which includes a definition of when a term is
$\mathbf{Safe}$, which is when it can never be caused to produce
blame.  We use a related technique to verify that modules cannot be
blamed, by constructing the \syntax{havoc} context.  However, we do
not attempt to construct a semantic model of contracts; instead we
merely approximate the run-time behaviors of programs with contracts.

\paragraph{Abstract interpretation}
Abstract interpretation provides a general theory of semantic
approximation \cite{dvanhorn:Cousot:1977:AI} that relates concrete
semantics to an abstract semantics that interprets programs over a domain of
abstract values.  Our approach is very much an instance of abstract
interpretation.  The reachable state semantics of CPCF is our
concrete semantics, with the semantics of SCPCF as an abstract interpretation
defined over the union of concrete values and abstract values
represented as sets of contracts.  In a first-order setting, contracts
have been used as abstract values \cite{dvanhorn:Fahndrich2011Static}.
Our work applies this idea to behavioral contracts and higher-order
programs.

\paragraph{Combining expressions with specifications}
 Giving semantics to programs combined with specifications has a long
 history in the setting of program
 refinements~\cite{dvanhorn:RalphJohan1998Refinement}.
Our key
 innovations are (a) treating specifications as abstract values,
 rather than as programs in a more abstract language, (b) applying abstract
 reduction to modular program analysis, as opposed to program
 derivation or by-hand verification, and (c) the use of higher-order contracts
 as specifications.

Type inference and checking can be recast as a reduction
semantics~\cite{dvanhorn:esop:kmf07}, and doing so
bears a conceptual resemblance to our contracts-as-values reduction.  The
principal difference is that Kuan et al.~are concerned with producing
a \emph{type}, and so all expressions are reduced to types
before being combined with other types.  Instead, we are concerned with
\emph{values}, and thus contracts are maintained as
specification values, but concrete values are not abstracted away.

Also related to our specification-as-values notion of reduction is
Reppy's~\cite{dvanhorn:Reppy2006Typesensitive} variant of 0CFA that
uses ``a more refined representation of approximate values'', namely
types.  The analysis is modular in the sense that all module imports
are approximated by their type, whereas our approach allows more
refined analysis whenever components are not opaque.  Reppy's analysis
can be considered as an instance of our framework by applying the
techniques of section~\ref{sec:machine} and thus could be derived from
the semantics of the language rather than requiring custom design.

\paragraph{Modular program analysis}
\citet{dvanhorn:Shivers:1991:CFA},
\citet{dvanhorn:Serrano1995Control}, and \citet{dvanhorn:ashley-dybvig-toplas98} address modularity (in
the sense of open-world assumptions of missing program components) by
incorporating a notion of an \emph{external} or \emph{undefined}
value, which is analogous to always using the abstract value $\opaque$
for unknown modules, and therefore allowing more
descriptive contracts can be seen as a refinement of the abstraction
on missing program components.

Another sense of the words \emph{modular} and \emph{compositional} is that program components
can be analyzed in isolation and whole programs can be analyzed by
combining these component-wise analysis
results. \citet{dvanhorn:Flanagan1997Effective} presents a set-based
analysis in this style for analyzing untyped programs, with many
similar goals to ours, but without considering specifications and
requiring the whole program before the final analysis is available.    Banerjee and
Jensen~\cite{dvanhorn:banerjee-jensen-mscs03,dvanhorn:banerjee-icfp97}
and
\citet{dvanhorn:Lee2002Proof} take similar
approaches to  type-based and 0CFA-style analyses, respectively.

\paragraph{Other approaches to higher-order verification}

\citeauthor{dvanhorn:Kobayashi2011Predicate}
\cite{dvanhorn:Kobayashi2011Predicate,samth:kobayashi-popl09} have recently
proposed approaches to verification of temporal properties of
higher-order programs based on model checking.  This work differs from
ours in four important respects.  First, it addresses temporal
properties while we focus on behavioral properties.  Second, it uses
externally-provided specifications, whereas we use contracts, which
programmers already add to their programs.  Third, and most
importantly, our system handles opaque components, while
model-checking approaches are whole-program.  Fourth, it operates on
higher-order recursion schemes, a computational model with less power
than CPCF, the basis of our development.

\citet{samth:liquid} present Liquid Types, an extension to the type
system of OCaml which incorporates dependent refinement types, and
automatically discharges the obligations using a solver.  This
naturally supports the encoding of some uses of contracts, but
restricts the language of refinements to make proof obligations
decidable.  We believe that a combination of our semantics with an
extension to use such a solver to decide the $\proves\mval\mcon$
relation would increase the precision and effectiveness of our
system.

\section{Conclusion}

We have presented a technique for verifying modular
higher-order programs with behavioral software contracts.  Contracts
are a powerful specification mechanism that are already used in
existing languages.  We have shown that by using contracts as abstract
values that approximate the behavior of omitted components, a
reduction semantics for contracts becomes a verification system.
Further, we can scale this system both to a rich contract language,
allowing expressive specifications, as well as to a computable
approximation for automatic verification derived directly from our
semantics.
Our central lesson is that abstract reduction semantics can turn the
semantics of a higher-order programming language with executable
specifications into a symbolic executor and modular verifier for those
specifications.



\paragraph{Acknowledgments:}
We are grateful to Phillipe Meunier for discussions of his prior work
and providing code for the prototype implementation of his system; to
Casey Klein for help with Redex; and to Christos Dimoulas for
discussions and advice.

\bibliographystyle{plainnat}
\bibliography{dvh-bibliography,sth-bibliography,greenberg}

\iffull

\clearpage
\appendix
\section{Auxiliary definitions}

This section presents the full definitions of the type systems,
metafunctions and relations, and machine transitions from the earlier
sections of the paper.

\subsection{Types for PCF with Contracts}
\label{sec:cpcf-types}

The type system for CPCF is entirely conventional and taken from \citet{dvanhorn:Dimoulas2011Contract}.

\begin{figure}[h]
\begin{displayfig}{CPCF Type System}
\begin{mathpar}
\inferrule{ }
          {\Gamma\vdash\mnum:\stypnum}

\inferrule{ }
          {\Gamma\vdash\strue:\stypbool}

\inferrule{ }
          {\Gamma\vdash\sfalse:\stypbool}

\inferrule{ }
          {\Gamma\vdash\szero:\starr\stypnum\stypbool}

\inferrule{ }
          {\Gamma\vdash\struep:\starr\stypbool\stypbool}

\inferrule{ }
          {\Gamma\vdash\sblm\mmodvar\mmodvaro:\mtyp}

\inferrule{\Gamma(\mvar) = \mtyp}
          {\Gamma\vdash\mvar:\mtyp}

\inferrule{\Gamma\vdash\mexp:\starr\mtyp{\mtyp'}
  \\
  \Gamma\vdash\mexpo:\mtyp}
  {\Gamma\vdash\sapp\mexp\mexpo:\mtyp'}

\inferrule{\Gamma;\mvar:\mtyp\vdash\mexp:\mtyp'}
          {\Gamma\vdash\stlam\mvar\mtyp\mexp:\starr\mtyp{\mtyp'}}

\inferrule{\Gamma;\mvar:\mtyp\vdash\mexp:\mtyp}
          {\Gamma\vdash\strec\mvar\mtyp\mexp}

\inferrule{\Gamma\vdash\mexp:\stypbool\\
  \Gamma\vdash\mexp_1:\mtyp\\
  \Gamma\vdash\mexp_2:\mtyp}
          {\Gamma\vdash\sif\mexp{\mexp_1}{\mexp_2}:\mtyp}

\inferrule{\Gamma\vdash\mcon:\stcon\mtyp\\
  \Gamma\vdash\mexp:\mtyp}
          {\Gamma\vdash\achksimple\mcon\mexp:\mtyp}

\inferrule{\Gamma\vdash\mexp:\starr\mtyp\stypbool}
          {\Gamma\vdash\sflat\mexp:\stcon\mtyp}

\inferrule{\Gamma\vdash\mcon:\stcon\mtyp\\
  \Gamma\vdash\mcono:\stcon{\mtyp'}}
          {\Gamma\vdash\stdep\mcon\mtyp\mvar\mcono:\stcon{\starr\mtyp{\mtyp'}}}
\end{mathpar}
\end{displayfig}
\end{figure}

\subsection{Types for Symbolic PCF with Contracts}
\label{sec:scpcf-types}

The type system extends straightforwardly to handle abstract values
labeled with their types.

\begin{display}{SCPCF Type System}
\begin{mathpar}
\inferrule{ }
          {\Gamma\vdash\topaque:\mtyp}

\inferrule{\Gamma\vdash:\mpreval:\mtyp \\
  \mcon\in\mconset\implies \Gamma\vdash\mcon:\stcon\mtyp}
          {\Gamma\vdash\with\mpreval\mconset:\mtyp}

\end{mathpar}
\end{display}

\subsection{Operations on concrete values}

The following rules define $\absdelta$ for concrete values, we refer
to the subset of $\absdelta$ that relates only concrete values as
$\plaindelta$.  We assume $\mval$, $\mvalo$ are concrete here.


\begin{figure}
\begin{displayfig}{Concrete Operations}
\begin{align*}
(\ssucc,\mnum,\mnum+1) &\in\plaindelta\\
(\scar,\vcons\mvalo\mval,\mvalo) &\in\plaindelta\\
(\scdr,\vcons\mvalo\mval,\mval) &\in\plaindelta\\
(\splus,\mnum,\mnumo,\mnum+\mnumo) &\in\plaindelta\\
(\sequalp,\mnum,\mnum,\strue) &\in \plaindelta \\
\mnum\not=\mnumo \implies (\sequalp,\mnum,\mnumo,\sfalse) &\in \plaindelta\\
(\mathsf{cons},\mvalo,\mval,\vcons\mvalo\mval)&\in \plaindelta\\
(\snatp,\mnum,\strue) &\in \plaindelta \\
\mval\not\in\mathcal{N}\implies(\snatp,\mval,\sfalse) &\in \plaindelta\\
\mval\in\{\strue,\sfalse\} \implies (\sboolp,\mval,\strue) &\in\plaindelta\\
\mval\not\in\{\strue,\sfalse\} \implies (\sboolp,\mval,\sfalse) &\in\plaindelta\\
(\semptyp,\sempty,\strue) &\in\plaindelta\\
\mval\neq\sempty \implies (\semptyp,\mval,\sfalse) &\in\plaindelta\\
(\sconsp,\vcons{\mvalo}{\mval},\strue) &\in\plaindelta \\
\mval' \not=\vcons{\mvalo}{\mval} \implies (\sconsp,\mval',\sfalse) &\in \plaindelta\\
(\sprocp,\sreclam\mvaro\mvar\mexp,\strue) &\in\plaindelta\\
\mval\neq\sreclam\mvaro\mvar\mexp\implies (\sprocp,\mval,\sfalse) &\in\plaindelta\\
(\sfalsep,\sfalse,\strue) &\in \plaindelta \\
\mval\not=\sfalse\implies (\sfalsep,\mval,\sfalse) &\in \plaindelta
\end{align*}
\end{displayfig}
\end{figure}

(We have omitted the rules producing blame for arity mismatch and undefined cases.)

\subsection{Operations on symbolic values}
\label{sec:absdelta}

The defintion of $\absdelta$ on symbolic values is presented in figure~\ref{fig:absdelta-full}.
We assume $\mval$, $\mvalo$ are abstract.  Further relations on
abstract values are presented in figures \ref{fig:contract-rel}, \ref{fig:more-more-prove} and \ref{fig:more-prove}.

\begin{figure}
\begin{display}[$\proves\mval\moppred$ \emph{and} $\refutes\mval\moppred$]{Value proves or refutes base predicate}
\begin{mathpar}
\inferrule{\plaindelta(\moppred,\mval) \ni \strue}
          {\proves\mval\moppred}

\inferrule{\proves\mcon\moppred}
          {\proves{\with\mval{\mconset\cup \{ \mcon \}}}\moppred}
\\
\inferrule{\plaindelta(\moppred,\mval) \ni \sfalse}
          {\refutes\mval\moppred}

\inferrule{\refutes\mcon\moppred}
          {\refutes{\with\mval{\mconset\cup \{ \mcon \}}}\moppred}
\end{mathpar}
\end{display}
\caption{Provability relations}
\label{fig:more-more-prove}
\end{figure}

\subsection{Machine}

The full CESK machine is given in figure \ref{fig:full-machine}.
\label{full-machine}

\begin{figure}[t]
\begin{display}{}
\begin{align*}
\proves\mval\moppred\implies (\moppred,\mval,\strue) &\in \absdelta\\
\refutes\mval\moppred\implies (\moppred,\mval,\sfalse) &\in \absdelta\\
\ambig\mval\moppred
\implies (\moppred,\mval,\with\opaque{\liftpred\sboolp}) &\in \absdelta\\
\proves\mval\snatp\implies(\ssucc,\mval,\with\opaque{\liftpred\snatp}) &\in \absdelta\\
\refutes\mval\snatp\implies(\ssucc^\mlab,\mval,\ablm\mlab\ssucc\mval\lambda\mval)
&\in\absdelta
\\
\ambig\mval\snatp\implies(\ssucc,\mval,\with\opaque{\liftpred\snatp})&\in\absdelta\\
\wedge\ (\ssucc^\mlab,\mval,\ablm\mlab\ssucc\mval\lambda\mval)&\in\absdelta
\\
\proves\mval\sconsp\implies(\scar,\mval,\projleft(\mval)) &\in \absdelta
\\
\refutes\mval\sconsp\implies(\scar^\mlab,\mval,\ablm\mlab\scar\mval\lambda\mval)
&\in \absdelta
\\
\ambig\mval\sconsp\implies (\scar,\mval,\projleft(\mval)) &\in\absdelta
\\
\wedge\ (\scar^\mlab,\mval,\ablm\mlab\scar\mval\lambda\mval)
 &\in\absdelta
\\
\proves\mval\sconsp\implies(\scdr,\mval,\projright(\mval)) &\in \absdelta
\\
\refutes\mval\sconsp\implies(\scdr^\mlab,\mval,\ablm\mlab\scdr\mval\lambda\mval)
&\in \absdelta
\\
\ambig\mval\sconsp\implies (\scdr,\mval,\projright(\mval)) &\in\absdelta
\\
\wedge\ (\scdr^\mlab,\mval,\ablm\mlab\scdr\mval\lambda\mval)
 &\in\absdelta
\\
\proves\mval\snatp \wedge \proves\mvalo\snatp
\implies\\
(\splus,\mval,\mvalo,\with\opaque{\liftpred\snatp}) &\in \absdelta
\\
\refutes\mval\snatp 
\implies
(\splus,\mval,\mvalo,\ablm\mlab\splus\mval\lambda\mval)
&\in\absdelta
\\
\refutes\mvalo\snatp 
\implies
(\splus,\mval,\mvalo,\ablm\mlab\splus\mvalo\lambda\mvalo)
&\in\absdelta
\\
\ambig\mval\snatp \wedge \ambig\mvalo\snatp
\implies\\
(\splus,\mval,\mvalo,\with\opaque{\liftpred\snatp}) 
&\in \absdelta
\\
\ambig\mval\snatp \implies
(\splus,\mval,\mvalo,\ablm\mlab\splus\mval\lambda\mval)
&\in\absdelta
\\
\ambig\mvalo\snatp \implies
(\splus,\mval,\mvalo,\ablm\mlab\splus\mvalo\lambda\mvalo)
&\in\absdelta
\\
\proves\mval\snatp \wedge \proves\mvalo\snatp
\implies\\
(\sequalp,\mval,\mvalo,\with\opaque{\liftpred\sboolp}) &\in \absdelta
\\
\refutes\mval\snatp 
\implies
(\sequalp,\mval,\mvalo,\ablm\mlab\sequalp\mval\lambda\mval)
&\in\absdelta
\\
\refutes\mvalo\snatp 
\implies
(\sequalp,\mval,\mvalo,\ablm\mlab\sequalp\mvalo\lambda\mvalo)
&\in\absdelta
\\
\ambig\mval\snatp \wedge \ambig\mvalo\snatp
\implies\\
(\sequalp,\mval,\mvalo,\with\opaque{\liftpred\sboolp}) 
&\in \absdelta
\\
\ambig\mval\snatp \implies
(\sequalp,\mval,\mvalo,\ablm\mlab\sequalp\mval\lambda\mval)
&\in\absdelta
\\
\ambig{\mval_1}\snatp \implies
(\sequalp,\mval,\mvalo,\ablm\mlab\sequalp\mvalo\lambda\mvalo)
&\in\absdelta
\\
(\sconsop,\mval,\mvalo,\vcons{\mval}{\mvalo})
&\in\absdelta
\end{align*}
\end{display}
\caption{Basic operations on abstract values}
\label{fig:absdelta-full}
\end{figure}

\begin{figure}
\begin{display}[$\proves\mcon\moppred$ \emph{and} $\refutes\mcon\moppred$]{Contract proves or refutes base predicate}
\begin{mathpar}
\inferrule{}
          {\proves{\liftpred\sfalsep}\sboolp}

\inferrule{}
          {\proves{\sconsc\mcon\mcono}\sconsp}

\inferrule{}
          {\proves{\sdep\mcon\mvar\mcono}\sprocp}

\inferrule{}
          {\proves{\liftpred\moppred}\moppred}          

\inferrule{\proves\mcon\moppred \\ \proves\mcono\moppred}
          {\proves{\sorc\mcon\mcono}\moppred}

\inferrule{\proves\mcon\moppred \mbox{ or } \proves\mcono\moppred}
          {\proves{\sandc\mcon\mcono}\moppred}
\\
\inferrule{\moppred \neq \sprocp}
            {\refutes{\sdep\mcon\mvar\mcono}\moppred}

\inferrule{\moppred \neq \sconsp}
          {\refutes{\sconsc\mcon\mcono}\moppred}

\inferrule{\refutes\mcon\moppred \\ \refutes\mcono\moppred}
          {\refutes{\sorc\mcon\mcono}\moppred}

\inferrule{\refutes\mcon\moppred \mbox{ or } \refutes\mcono\moppred}
          {\refutes{\sandc\mcon\mcono}\moppred}

\inferrule{\refutes{\subst{\srecc\mvar\mcon}\mvar\mcon}\moppred}
          {\refutes{\srecc\mvar\mcon}\moppred}

\inferrule{\moppred\neq\moppred' \\ 
  \{\moppred,\moppred'\}\neq\{\sfalsep,\sboolp\}}
          {\refutes{\liftpred{\moppred'}}\moppred}
\end{mathpar}
\end{display}
\caption{Contract relations}
\label{fig:contract-rel}
\end{figure}





\begin{figure*}
\begin{machine}[$\s\mexp\menv\msto\mcont \stdstep \s\mexpo\menvo\mstoo\mconto$]{Basic reductions}
\se{\sapp\mexp\mexpo^\mlab}
&\mstep&
\s\mexp\menv{\msto[\mkaddr\mapsto\mcont]}{\kap\mexpo\menv\mlab\mkaddr}
\\
\se{\sif\mexp{\mexpo_1}{\mexpo_2}}
&\mstep&
\s\mexp\menv{\msto[\mkaddr\mapsto\mcont]}{\kif{\mexpo_1}{\mexp_2}\menv\mkaddr}
\\
\se{\sapp\mop\mexp^\mlab}
&\mstep&
\s\mexp\menv{\msto[\mkaddr\mapsto\mcont]}{\kopone\mop\mlab\mkaddr}
\\
\se{\sapp\mop{\mexp\;\mexpo}^\mlab}
&\mstep&
\s\mexp\menv{\msto[\mkaddr\mapsto\mcont]}{\koptwol\mop\mexpo\menv\mlab\mkaddr}
\\
\se\mvar
&\mstep&
\s\mval{\menvo}\msto\mcont
&
\mbox{ if }(\mval,\menvo)\in\msto(\menv(\mvar))
\\
\sk{\kap\mexp\menvo\mlab\mkaddr}
&\mstep&
\s\mexp\menvo{\msto[\mkaddr\mapsto\mcont]}{\kfn\mval\menv\mlab\mkaddr}
\\
\sk{\kfn{\slam\mvar\mexp}{\menvo}\mlab{\mkaddr}}
&\mstep&
\multicolumn{2}{l}{
$\s\mexp{\menvo[\mvar\mapsto\maddr]}{\msto[\maddr\mapsto(\mval,\menv)]}\mcont$
}
\\
\sk{\kfn{\mvalo}{\menvo}\mlab{\mkaddr}}
&\mstep&
\s{\simpleblm{\mlab}{\Lambda}}\emptyset\emptyset\kmt
& 
\mbox{ if }\deltamap\sprocp\mvalo\sfalse
\\
\sk{\kif\mexp\mexpo{\menvo}\mkaddr}
&\mstep&
\s\mexp{\menvo}\msto\mcont 
&
\mbox{ if }\deltamap\sfalsep\mval\sfalse
\\
\sk{\kif\mexp\mexpo{\menvo}\mkaddr}
&\mstep&
\s\mexpo{\menvo}\msto\mcont 
&
\mbox{ if }\deltamap\sfalsep\mval\strue
\\
\sk{\kopone\mop\mlab\maddr}
&\mstep&
\s\mans\emptyset\msto\mcont
&
\mbox{ if }\deltamap{\mop^\mlab}\mval\mans
\\
\s{\vcons{(\mvalo,\menvo)}{(\mval,\menv)}}\emptyset\msto{\kopone\scar\relax\maddr}
&\mstep&
\s\mvalo\menvo\msto\mcont
\\
\s{\vcons{(\mvalo,\menvo)}{(\mval,\menv)}}\emptyset\msto{\kopone\scdr\relax\maddr}
&\mstep&
\s\mval\menv\msto\mcont
\\
\sk{\koptwol\mop\mexp\menvo\mlab\mkaddr}
&\mstep&
\s\mexp\menvo\msto{\koptwo\mop\mval\menv\mlab\mkaddr}
\\
\sk{\koptwo\sconsop\mvalo\menvo\mlab\maddr}
&\mstep&
\s{\vcons{(\mvalo,\menvo)}{(\mval,\menv)}}\emptyset\msto\mcont
\\
\sk{\koptwo\mop\mvalo\menvo\mlab\maddr}
&\mstep&
\s\mans\emptyset\msto\mcont
\\
\s{\simpleblm\mlab\mlabo}\menv\msto\mcont
&\mstep&
\s{\simpleblm\mlab\mlabo}\emptyset\emptyset\kmt
\end{machine}

\begin{machine}{Module references}
\s{\mmodvar^\mmodvar\!}\menv\msto\mcont
&\mstep&
\s\mval\emptyset\msto\mcont
&
\mbox{ if }(\mmodvar^\mmodvar\!, \mval)\in\widehat\Delta(\vec\mmod)
\\
\s{\mmodvar^\mmodvaro\!}\menv\msto\mcont
&\mstep&
\s\mval\emptyset{\msto[\mkaddr\mapsto\mcont]}{\kchk\mcon\emptyset\mmodvar\mmodvaro\mmodvaroo\mkaddr}
&
\mbox{ if }(\mmodvar^\mmodvar\!, \chk\mcon\mmodvar\mmodvaro\mmodvaroo\mval)\in\widehat\Delta(\vec\mmod)
\end{machine}

\begin{machine}{Contract checking}
\s{\chk\mcon\mmodvar\mmodvaro\mmodvaroo\mexp}\menv\msto\mcont
&\mstep&
\s\mexp\menv{\msto[\mkaddr\mapsto\mcont]}{\kchk\mcon\menv\mmodvar\mmodvaro\mmodvaroo\mkaddr}
\\
\sk{\kchk\mcon\menvo\mmodvar\mmodvaro\mmodvaroo\mkaddr}
&\mstep&
\s\mval\menv{\msto[\mkaddr'\mapsto\kif{\with\mval{\{\mcon\}}}{\simpleblm{\mmodvar}{\mmodvaro}}\menv\mkaddr]}{\kfn\mvalo\menvo\relax{\mkaddr'}}
\\
&&
\mbox{ where }\mcon \mbox{ is flat}
\mbox{ and }\mvalo=\textsc{fc}(\mcon,\mval)
\\
\sk{\kfn{\depbless\mcon\mvar\mcono\mmodvar\mmodvaro\mmodvaroo\maddr}{\menvo}\mlab\mkaddr}
\\
&\mstep&
\langle{\mval},{\menv},{\kchk{\mcon}{\menvo}\mmodvaro\mmodvar\mmodvaroo{\mkaddr'}},
\msto[
\begin{tabular}[t]{@{}>{$}l<{$}}
\mkaddr' \mapsto \kfn{\mvalo}{\menvo'\!}\mlab{\mkaddr''},\\
\mkaddr'' \mapsto \kchk{\mcono}{\menvo[\mvar\mapsto\maddro]}\mmodvar\mmodvaro\mmodvaroo\mkaddr,\\
\maddro\mapsto(\mval,\menv)]\rangle
\end{tabular}
\\
&&\mbox{where } (\mvalo,\menvo')\in\msto(\maddr)
\\
\sk{\kchk{\sdep\mcon\mvar\mcono}{\menvo}\mmodvar\mmodvaro\mmodvaroo\mkaddr}
&\mstep&
\s{\depbless\mcon\mvar\mcono\mmodvar\mmodvaro\mmodvaroo\maddr}{\menvo}{\msto[\maddr\mapsto(\mval,\menv)]}\mconto
&
\\
&&&
\mbox{ if }\deltamap\sprocp\mval\strue
\\
\sk{\kchk{\sdep\mcon\mvar\mcono}\menvo\mmodvar\mmodvaro\mmodvaroo\mkaddr}
&\mstep&
\s{\simpleblm\mmodvar\mmodvaroo}\emptyset\emptyset\kmt
&
\mbox{ if }\deltamap\sprocp\mval\sfalse
\\
\sk{\kchk{\sandc\mcon\mcono}\menvo\mmodvar\mmodvaro\mmodvaroo\mkaddr}
&\mstep&
\s\mval\menv{\msto[\mkaddro\mapsto\kchk\mcono\menvo\mmodvar\mmodvaro\mmodvaroo\mkaddr]}{\kchk\mcon\menvo\mmodvar\mmodvaro\mmodvaroo\mkaddro}
\\
\sk{\kchk{\sorc\mcon\mcono}\menvo\mmodvar\mmodvaro\mmodvaroo\mkaddr}
&\mstep&
\s\mval\menv{\msto[\mkaddro\mapsto\kchkor\mval\menv{\sorc\mcon\mcono}\menvo\mmodvar\mmodvaro\mmodvaroo\mkaddr]}{\kap\mvalo\menvo\relax\mkaddro}
\\
&&\mbox{where }\mvalo=\fc(\mcon)
\\
\sk{\kchkor\mvalo\menvo{\sorc\mcon\mcono}{\menv'}\mmodvar\mmodvaro\mmodvaroo\mkaddr}
&\mstep&
\s{\with\mvalo{\{\mcon\}}}\menvo\msto\mcont 
&
\mbox{ if }\deltamap\sfalsep\mval\sfalse
\\
\sk{\kchkor\mvalo\menvo{\sorc\mcon\mcono}{\menv'}\mmodvar\mmodvaro\mmodvaroo\mkaddr}
&\mstep&
\s\mvalo\menvo\msto{\kchk\mcono{\menv'\!}\mmodvar\mmodvaro\mmodvaroo\mkaddr}
&
\mbox{ if }\deltamap\sfalsep\mval\strue
\end{machine}

\begin{machine}{Abstract values}
\sk{\kfn{\with\opaque\mconset}\menvo\mlab\mkaddr}
&\mstep&
\s\mexp\menv\msto{\kdem\mvalo\menvo\mkaddr}
&
\mbox{ if }\deltamap\sprocp{\with\opaque\mconset}\strue
\\
&&\multicolumn{2}{l}{$\mbox{where }\mexp=\textsc{amb}(\{\strue,\textsc{demonic}(\bigwedge\textsc{dom}(\mconset),\mval)\})\mbox{ and }\mvalo=\with\opaque{\textsc{rng}(\mconset)}$}
\\
\sk{\kdem\mexp\menvo\mkaddr}
&\mstep&
\s\mexp\menvo\msto\mcont
\\
\s{\with\opaque{\mconset\dotcup\{\sorc{\mcon_1}{\mcon_2}\}}}\menv\msto\mcont
&\mstep&
\s{\with\opaque{\mconset\cup\{\mcon_i\}}}\menv\msto\mcont
\\
\s{\with\opaque{\mconset\dotcup\{\srecc\mvar\mcon\}}}\menv\msto\mcont
&\mstep&
\s{\with\opaque{\mconset\dotcup\{\subst{\srecc\mvar\mcon}\mvar\mcon\}}}\menv\msto\mcont
\end{machine}

\begin{machine}{Higher-order pair contract checking}
\s\mval\menv\msto{\kchk{\sconsc\mcon\mcono}{\menvo}\mmodvar\mmodvaro\mmodvaroo\mkaddr}
&\mstep&
\s{\simpleblm\mmodvar\mmodvaroo}\emptyset\emptyset\kmt
&\mbox{if }\deltamap\sconsp\mval\sfalse
\\
\s\mval\menv\msto{\kchk{\sconsc\mcon\mcono}{\menvo}\mmodvar\mmodvaro\mmodvaroo\mkaddr}
&\mstep&
\s\mvalo\menv{\msto[\mkaddro\mapsto\kchk\mcon\menvo\mmodvar\mmodvaro\mmodvaroo{\mkaddr'},\mkaddr'\mapsto\kchkconso\mcono\menvo\mvalo\menv\mmodvar\mmodvaro\mmodvaroo\mkaddr]}{\kopone\scar\relax\mkaddro}
\\
&&
\multicolumn{2}{l}{
$\mbox{if }\deltamap\sconsp\mval\strue,
\mbox{ where }\mvalo=\with\mval\{\liftpred\sconsp\}$}
\\
\sk{\kchkconso\mcon\menvo\mvalo{\menv'}\mmodvar\mmodvaro\mmodvaroo\mkaddr}
&\mstep&
\s\mvalo{\menv'}{\msto[\mkaddro\mapsto\kchk\mcon\menvo\mmodvar\mmodvaro\mmodvaroo{\mkaddr'},
\mkaddr'\mapsto\koptwo\sconsop\mval\menv\relax\mkaddr]}
{\kopone\scdr\relax\mkaddro}
\end{machine}
\caption{Full CESK$^*$ machine}
\label{fig:full-machine}
\end{figure*}

\begin{figure}
\begin{display}[$\proves\mval\mcon$ \emph{and} $\refutes\mval\mcon$]{Value proves or refutes contract}
\begin{mathpar}
\inferrule{\mcon \in \mconset}
          {\proves{\with\mval\mconset}{\mcon}}

\inferrule{\proves\mval\moppred}
          {\proves\mval{\liftpred\moppred}}
\\
\inferrule{\refutes\mval\moppred}
          {\refutes\mval{\liftpred\moppred}}

\inferrule{\refutes\mval\mcon \\ \refutes\mval\mcono}
          {\refutes\mval{\sorc\mcon\mcono}}

\inferrule{\refutes\mval\mcon \mbox{ or } \refutes\mval\mcono}
          {\refutes\mval{\sandc\mcon\mcono}}

\inferrule{\mcon \in \{\liftpred\sconsp,\liftpred\snatp,\liftpred\sfalsep,\liftpred\sboolp\} \\ \proves\mval\sprocp}
          {\refutes\mval\mcon}

\inferrule{\refutes\mval\sprocp}
          {\refutes\mval{\sdep\mcon\mvar\mcono}}

\inferrule{\refutes\mval\sconsp}
          {\refutes\mval{\sconsc\mcon\mcono}}

\inferrule{\refutes{\projleft(\mval)}\mcon \mbox{ or } 
           \refutes{\projright(\mval)}\mcono}
          {\refutes\mval{\sconsc\mcon\mcono}}

\inferrule{\refutes\mvalo\mcon \mbox{ or } 
           \refutes\mval\mcono}
          {\refutes{\vcons\mvalo\mval}{\sconsc\mcon\mcono}}

\inferrule{\refutes\mval{\subst{\srecc\mvar\mcon}\mvar\mcon}}
          {\refutes\mval{\srecc\mvar\mcon}}
\end{mathpar}
\end{display}
\caption{Provability relations}
\label{fig:more-prove}
\end{figure}

\clearpage

\section{Proofs}
\label{sec:proofs}

The approximation relation on expressions, modules, and programs is
formalized below.  We show only the important cases and omit the
straightforward structurally recursive cases.
We parametrize by the module context of the abstract program,
to determine the opaque modules; we omit this context where it can be
inferred from context.
\begin{display}[$\mprg \sqsubseteq \mprgo$, $\mmod \sqsubseteq_\mvmod
  \mmodo$, \emph{and} $\mexp \sqsubseteq_\mvmod \mexpo$]{Approximates}
\begin{mathpar}
\inferrule{ }
          {\mval\sqsubseteq\opaque}

\inferrule{ }
          {\achksimple\mcon\mexp \sqsubseteq \opaque\cdot\mcon}

\inferrule{ }
          {\slam\mvar{\achksimple\mcono{(\sapp\mval{\achksimple\mcon\mvar})}}
            \sqsubseteq
            \opaque\cdot\sdep\mcon\mvar\mcono}

\inferrule{ }
          {\mval\cdot\mcon\sqsubseteq\mval}

\inferrule{\mval\sqsubseteq\mval'}
          {\mval\cdot\mcon\sqsubseteq\mval'\cdot\mcon}

          \inferrule{\proves\mval\mcon}
          {\mval\sqsubseteq\mval\cdot\mcon}
          
\inferrule{\mvmodo\sqsubseteq_\mvmod \mvmod \\ \mexpo\sqsubseteq_\mvmod \mexp}
          {\mvmodo\mexpo \sqsubseteq \mvmod\mexp}

          \inferrule
          {\achksimple\mcon\mexp \sqsubseteq \mexpo}
          {\achksimple\mcon\mexp \sqsubseteq \achksimple\mcon\mexpo}
          
  \inferrule{\smod\mmodvar\mcon\bullet \in \vec\mmod \text{ or } \mmodvar=\dagger}{
     \ablm{\mmodvar}{\mmodvaro}{\mval}{\mcon'}{\mval'} \sqsubseteq_\mvmod \mexp} 

\inferrule{\smod\mmodvar\mcon\opaque \in \mvmod}
          {\smod\mmodvar\mcon\mval \sqsubseteq_\mvmod \smod\mmodvar\mcon\opaque}
\end{mathpar}
\end{display}
We lift $\sqsubseteq$ to evaluation contexts $\mctx$ by structural
extension.  We lift $\sqsubseteq$ to contracts by structural extension
on contracts and $\sqsubseteq$ on embedded values.  We lift
$\sqsubseteq$ to vectors by point-wise extension and to sets
of expressions by point-wise, subset extension.

With the approximation relation in place, we now state and prove our
main soundness theorem.
%



\begin{lemma}
\label{lem:delta}
If $\vec\mval\sqsubseteq_\mvmod\vec\mvalo$, then $\delta(\mop,\vec\mval) \sqsubseteq_\mvmod
\delta(\mop,\vec\mvalo)$.
\end{lemma}
\begin{proof}
By inspection of $\delta$ and cases on $\mop$ and $\vec\mval$.
\end{proof}


\begin{lemma}
\label{lem:subst}
If $\mexp\sqsubseteq_\mvmod\mexpo$ and $\mval\sqsubseteq_\mvmod\!\mvalo$, then
$\subst\mval\mvar\mexp \sqsubseteq_\mvmod \subst\mvalo\mvar\mexpo$.
\end{lemma}
\begin{proof}
  By induction on the structure of $\mexp$ and cases on the derivation
  of $\mexp\sqsubseteq_\mvmod\mexpo$.
\end{proof}

\begin{lemma}
\label{lem:fc}
If $\mcon\sqsubseteq_\mvmod\mcono$, then $\fc(\mcon) \sqsubseteq_\mvmod \fc(\mcono)$.
\end{lemma}
\begin{proof}
  By induction on the structure of $\mcon$ and cases on the derivation
  of $\mexp\sqsubseteq_\mvmod\mexpo$ and the defintion of $\fc$.
\end{proof}


\begin{lemma}
\label{lem:proves-fc}
Let $\mexp = \fc(\mcon)$, then
\begin{enumerate}
\item if $\proves\mval\mcon$ and $\mval\sqsubseteq_\mvmod\mvalo$, then
$\sapp\mexp\mvalo \multistdstep \mans \sqsupseteq \mbox{\emph{\strue}}$,
\item if $\refutes\mval\mcon$ and $\mval\sqsubseteq_\mvmod\mvalo$, then
$\sapp\mexp\mvalo \multistdstep \mans \sqsupseteq \mbox{\emph{\sfalse}}$.
\end{enumerate}
\end{lemma}

\begin{proof}
  By induction on the structure of $\mcon$ and cases on the derivation
  of $\mvalo\sqsubseteq_\mvmod\mval$ and the defintion of $\fc$.
\end{proof}

\begin{lemma}
If $\mprg\;\stdstep\;\mprg'$,
\emph{$\mprg'\neq\mvmod\:\simpleblm\mmodvar\mmodvaro$} and
$\mprg\sqsubseteq\mprgo$, then
$\mprgo\;\multistdstep\;\mprgo'$ and $\mprg'\sqsubseteq\mprgo'$.
\label{lem:main-lemma}
\end{lemma}
\begin{proof}
We split into two cases.

\noindent
Case (1):
\begin{align*}
\mprg &= \mvmod\;\mctx[\mexp] \stdstep \mvmod\;\mctx[\mexp']\\
\mprgo &= \vec\mmodo\;\mctx'[\mexpo] \stdstep \vec\mmodo\;\mctx'[\mexpo']
\end{align*}
where $\mctx\sqsubseteq_\mvmodo\mctx'$.  We reason by cases on the step
from $\mexp$ to $\mexp'$.
\begin{itemize}

\item Case: $\mexp = \mmodvar^\mmodvaro$ and $\smod\mmodvar\mcon{\mexp'}
  \in \mvmod$

If $\mmodvar$ is transparent in $\vec\mmodo$, then
$\smod\mmodvar\mcon{\mexp'}
  \in \vec\mmodo$  and we are done by simple application of the
  reduction rules for module references.
Otherwise, $\mmodvar\not=\mmodvaro$ and thus
\begin{align*}
\mexp' &= \chk\mcon\mmodvar\mmodvaro\mmodvar\mval &
\mexpo'&= \chk\mcon\mmodvar\mmodvaro\mmodvar{\with\opaque{\{\mcon\}}}\text,
\end{align*}
but now $\mexp' \sqsubseteq_\mvmodo \mexpo'$, since $\mexp'\sqsubseteq_\mvmodo
{\with\opaque{\{\mcon\}}}$.

\item Case: $\achksimple{\sdep\mcon\mvar\mcono}\mval
  \;\stdstep\;
\simpledepbless\mcon\mvar\mcono{\mvalo}$
where $\deltamap\sprocp\mval\strue$
 and $\mvalo = \with\mval{\{\sdep\mcon\mvar\mcono\}}$.

Since $\mexpo$ is a redex, by $\sqsubseteq$ we have $\mexpo=
\achksimple{\sdep{\mcon'}\mvaro{\mcono'}}{\mval'}$, where
$\mval\sqsubseteq_\mvmodo\mval'$.  By lemma~\ref{lem:delta},
$\deltamap\sprocp{\mval'}\strue$.  So
$\mexpo'=\simpledepbless{\mcon'}\mvaro{\mcono'}{\mvalo'}$ where
$\mvalo'=\with{\mval'}{\{\sdep{\mcon'}\mvaro{\mcon'}\}} \sqsubseteq_\mvmodo
\with\mval{\{\sdep\mcon\mvar\mcono\}}$ and thus
$\mexp'\sqsubseteq_\mvmodo\mexpo'$.

\item
  Case: $\sapp{\mval_1}{\mval_2}^\mlab\;\stdstep\;
  \asblm{\mlab}{\Lambda}{\mval}$, where
  $\deltamap\sprocp{\mval_1}\sfalse$.

By $\sqsubseteq$, we have $\mexpo = \sapp{\mvalo_1}{\mvalo_2}^\mlab$
and $\mvalo_i \sqsubseteq_\mvmodo \mval_i$.  By lemma~\ref{lem:delta},
$\deltamap\sprocp{\mvalo_1}\sfalse$, hence
$\sapp{\mvalo_1}{\mvalo_2}^\mlab\;\stdstep\; \simpleblm{\mlab}{\Lambda}$.

\item
  Case: $\sapp{\mval_1}{\mval_2}^\mlab\;\stdstep\;
  \mexp'$, where
  $\deltamap\sprocp{\mval_1}\strue$.
 
By $\sqsubseteq$, we have $\mexpo=\sapp{\mvalo_1}{\mvalo_2}^\mlab$ and
  $\mvalo_i \sqsubseteq_\mvmodo \mval_i$.  By lemma~\ref{lem:delta},
$\deltamap\sprocp{\mvalo_1}\strue$.

Either $\mval_1$ and $\mvalo_1$ are structurally similar, in which
case the result follows by possibly relying on lemma~\ref{lem:subst},
or $\mval_1 = \with{\sreclam\mvar\mvaro{\mexp_0}}\mconset$ and
$\mvalo_1 = \with\opaque{\mconset'}$.  There are two possibilities for
the origin of $\mval_1$: either it was blessed or it was not.  If
$\mval_1$ was not blessed, then $\mconset$ contains no function
contracts, implying $\mconset'$ contains no function contracts, hence
$\mexpo'= \opaque$,
and the result holds.  Alternatively, $\mval_1$ was blessed and
$\mconset$ contains a function contract $\sdep\mcon\mvar\mcono$.  But
but by the blessed application rule, we have $\mctx =
\mctx_1[\achksimple{\subst{\mval'_2}\mvar\mcono}{[\;]}]$, thus by
assumption $\mctx'=
\mctx'_1[\achksimple{\subst{\mvalo'_2}\mvar{\mcono'}}{[\;]}]$,
implying $\subst{\mval'_2}\mvar\mcono \sqsubseteq
\subst{\mvalo'_2}\mvar{\mcono'}$, finally giving us the needed
conclusion:
\[
\mctx_1[\achksimple{\subst{\mval'_2}\mvar\mcono}{\mexp'}]
\sqsubseteq_\mvmodo
\mctx'_1[\achksimple{\subst{\mvalo'_2}\mvar{\mcono'}}{\mexpo'}]\text,
\]
where $\mexpo'=\with\opaque\{\subst{\mvalo_2}\mvar{\mcono'}\ |\ \sdep{\mcon'}\mvar{\mcono'} \in \mconset'\}$. 

\item Case: $\achksimple\mcon\mval\;\stdstep\;\mexp'$ where $\mcon$ is flat.

If $\ambig\mval\mcon$, then the case holds by use of lemma~\ref{lem:fc}.
If $\proves\mval\mcon$, then the case holds by lemma~\ref{lem:proves-fc}(1).
If $\refutes\mval\mcon$, then the case holds by lemma~\ref{lem:proves-fc}(2).

\item The remaining cases are straightforward.
\end{itemize}

\noindent
Case (2): 
\begin{align*}
\mprg &= \mvmod\;\mctx_1[\mctx_2[\mexp]] \stdstep
  \mvmod\;\mctx_1[\mctx_2[\mexp']]\\
\mprgo &=
  \vec\mmodo\;\mctx'_1[\mctx'_2[\mexpo]]
\end{align*}
where $\mctx_1$ is the largest context such that
$\mctx_1\sqsubseteq_\mvmodo\mctx'_1$ but $\mctx_2 \not\sqsubseteq_\mvmodo \mctx'_2$.

In this case, we have $\mctx_2[\mexp] \sqsubseteq_\mvmodo
\mctx'_2[\mexpo]$, but since $\mctx_2\not\sqsubseteq_\mvmodo\mctx'_2$,
this must follow by one of the non-structural rules for $\sqsubseteq$,
all of which are either oblivious to the contents of $\mexp$ and
$\mexp'$, or do not relate redexes to anything.
\end{proof}


\begin{lemma}
If there exists a context $\mctx$ such that
\[
\mbox{\emph{$\mvmod\;\smod\mmodvar\mcon\mval\;\mctx[\mmodvar] \multistdstep
\simpleblm\mmodvar\mmodvaro$,}}
\]
then 
\[
\mbox{\emph{$\mvmod\;\smod\mmodvar\mcon\mval\;(\sapp{\text{\hv}}\mmodvar)\multistdstep\simpleblm\mmodvar\mmodvaro$.}}
\]
\label{lem:demonic}
\end{lemma}
\begin{proof}
If there exists such an $\mctx$, then 
without loss of generality, it is of some minimal form $\mathcal{D}$ in 
\begin{align*}
  \mathcal{D} &= [\;]\ |\ (\mathcal{D}\ V)\ |\ (\scar\ \mathcal{D})\ |\ (\scdr\ \mathcal{D})\text,
\end{align*}
but then there exists a $\mathcal{D}'$ equal to $\mathcal{D}$
with all values replaced with $\opaque$ such that
   $\mvmod\;\mathcal{D}'[\mval] 
   \multistdstep \simpleblm\mmodvar\mmodvaro$.
This is because at every reduction step, replacing some component of
the redex with $\opaque$ causes at least that reduction to fire, 
possibly in addition to others.
Further, by inspection of \hv, if  
   $\mvmod\;\mathcal{D}'[\mval] 
   \multistdstep \simpleblm\mmodvar\mmodvaro$,
then 
$\mvmod\;(\hv\ \mval) \multistdstep \simpleblm\mmodvar\mmodvaro$.
\end{proof}

\begin{proof}[Theorem~\ref{thm:soundness}]
By the definition of $\mathit{eval}$, we have
$\mprg\multistdstep\mans$.  
Let the number of steps in  $\mprg\multistdstep\mans$ be $n$.
There are two cases: either $\mans = \mval$, or $\mans =
\simpleblm\mlab\mlabo$.  If $\mans = \mval$, then we proceed by
induction on $n$ and  apply lemma~\ref{lem:main-lemma} at each step.  

If $\mans = \simpleblm\mlab\mlabo$ then there are two possibilities.
If $\mlab$ is the name of an opaque module in $\mvmod$ or if $\mlab=\dagger$, then
 $\mans \sqsubseteq_\mvmod \manso$ immediately.  If $\mlab=\mmodvar$ is the name
of a concrete module in $\mvmod$, then
$\sapp{\hv}\mmodvar \multistdstep \mans$ by
lemma~\ref{lem:demonic}, and therefore $\mans \in
\mathit{eval}(\mprgo)$ by the definition of $\mathit{eval}$. 
\end{proof}

\begin{proof}[Theorem~\ref{thm:soundness-machine}]
We reason by case analysis on the transition. In the cases where the
transition is deterministic, the result follows by calculation. For
the the remaining non-deterministic cases, we must show an abstract
state exists such that the simulation is preserved. By examining the
rules for these cases, we see that all hinge on the abstract store in
$\hat\msto$ soundly approximating the concrete store in $\msto$, which
follows from the assumption that $\alpha(\msto) \sqsubseteq
\hat\msto$.
\end{proof}

\begin{proof}[Theorem~\ref{thm:decidability}]
The state-space of the machine is non-recursive with finite sets at
the leaves on the assumption that addresses and base values are
finite. Hence reachability is decidable since the abstract state-space
is finite.
\end{proof}

\else
\fi

\end{document}